\documentclass[11pt]{article}

\usepackage{booktabs} 
\usepackage[ruled,linesnumbered,noend]{algorithm2e} 
\SetKwProg{Fn}{function}
\SetAlFnt{\relax}
\SetAlCapFnt{\relax}
\SetAlCapNameFnt{\relax}
\SetAlCapHSkip{0pt}

\let\MR\relax

\usepackage{mathtools}

\usepackage{amsmath}
\usepackage{amssymb}
\usepackage{amsthm}
\usepackage{bm}
\usepackage{xcolor,colortbl}
\usepackage{dsfont}
\usepackage{authblk}
\usepackage{nicefrac}
\usepackage{complexity}

\usepackage{tikz}
\usepackage{forest}
\usetikzlibrary{shapes.misc,positioning,matrix,arrows.meta}
\usepackage{bbm}
\usepackage{complexity}
\usepackage{float}
\usepackage{mathrsfs}

\usepackage{pgfplots}
\pgfplotsset{compat=newest}

\usepackage[colorlinks,citecolor=blue]{hyperref}
\usepackage{mleftright}
\usepackage{thm-restate}
\usepackage{physics}
\usepackage[capitalize,noabbrev,nameinlink]{cleveref}
\usepackage[shortlabels]{enumitem}
\usepackage{xspace}

\usepackage{fullpage}
\usepackage{parskip}
\usepackage{MnSymbol}

\crefname{item}{Item}{Item}

\newcounter{qst}
\crefname{qst}{Question}{Questions}

\newcommand{\declarecolor}[2]{\definecolor{#1}{RGB}{#2}\expandafter\newcommand\csname #1\endcsname[1]{\textcolor{#1}{##1}}}
\declarecolor{White}{255, 255, 255}
\declarecolor{Black}{0, 0, 0}
\declarecolor{Maroon}{128, 0, 0}
\declarecolor{Coral}{255, 127, 80}
\declarecolor{Red}{182, 21, 21}
\declarecolor{LimeGreen}{50, 205, 50}
\declarecolor{DarkGreen}{0, 80, 0}
\declarecolor{Purple}{146, 42, 158}
\declarecolor{Navy}{0, 0, 128}
\declarecolor{LightBlue}{84, 101, 202}
\definecolor{mydarkblue}{rgb}{0,0.08,0.45}
\hypersetup{ %
    pdftitle={},
    pdfkeywords={},
    pdfborder=0 0 0,
    pdfpagemode=UseNone,
    colorlinks=true,
    linkcolor=Maroon,
    citecolor=Navy,
    filecolor=Purple,
    urlcolor=Purple,
}

\makeatletter
\patchcmd\algocf@Vline{\vrule}{\vrule \kern-0.4pt}{}{}
\patchcmd\algocf@Vsline{\vrule}{\vrule \kern-0.4pt}{}{}
\makeatother

\makeatletter
\let\cref@old@stepcounter\stepcounter
\def\stepcounter#1{%
  \cref@old@stepcounter{#1}%
  \cref@constructprefix{#1}{\cref@result}%
  \@ifundefined{cref@#1@alias}%
    {\def\@tempa{#1}}%
    {\def\@tempa{\csname cref@#1@alias\endcsname}}%
  \protected@edef\cref@currentlabel{%
    [\@tempa][\arabic{#1}][\cref@result]%
    \csname p@#1\endcsname\csname the#1\endcsname}}
\makeatother

\theoremstyle{plain}
\newtheorem{theorem}{Theorem}[section]
\newtheorem*{theorem*}{Theorem}
\newtheorem{lemma}[theorem]{Lemma}
\newtheorem{corollary}[theorem]{Corollary}
\newtheorem{proposition}[theorem]{Proposition}

\newtheorem{claim}[theorem]{Claim}

\theoremstyle{definition}
\newtheorem{definition}[theorem]{Definition}

\theoremstyle{remark}
\newtheorem{remark}[theorem]{Remark}

\let\E\relax

\newcommand{\tlow}[1]{\underline{t_{#1}}}
\newcommand{\thigh}[1]{\overline{t_{#1}}}

\newcommand{\vx}{\vec{x}}

\renewcommand{\vu}{\vec{u}}

\newcommand{\vxstar}{\vec{x}^\star}

\newcommand{\proj}{\Pi}

\usepackage{natbib}
\newcommand*{\N}{{\mathbb{N}}}

\let\R\relax
\newcommand*{\R}{{\mathbb{R}}}
\newcommand*{\E}{\operatornamewithlimits{\mathbb E}}

\newcommand*{\cX}{{\mathcal{X}}}

\newcommand*{\cR}{{\mathcal{R}}}

\newcommand*{\cA}{{\mathcal{A}}}

\newcommand{\defeq}{\coloneqq}

\newcommand{\tcrit}{t_{c}}

\newcommand{\BRgap}{\mathsf{BRGap}}
\newcommand{\KKTgap}{\mathsf{KKTGap}}
\newcommand{\Phimax}{\Phi_{\mathsf{range}}}

\newcommand{\reg}{\mathsf{Reg}}

\mathchardef\mathhyphen="2D

\NewDocumentCommand{\treeset}{o}{\mathbb{T}\IfNoValueF{#1}{_{#1}}}

\newcommand{\mwu}{\texttt{MWU}}

\newcommand{\vR}{\vec{r}}
\newcommand{\vtheta}{\vec{\theta}}
\newcommand{\vdelta}{\vec{\delta}}

\newcommand{\urange}{u_{\textsf{range}}}

\newcommand{\RM}{\texttt{RM}}
\newcommand{\RMplus}{\texttt{RM}^+}
\newcommand{\DRMplus}{\texttt{DRM}^+}

\newcommand{\md}{\texttt{MD}}
\newcommand{\ftrl}{\texttt{FTRL}}

\DeclarePairedDelimiterX{\infdivx}[2]{(}{)}{%
  #1\;\delimsize\|\;#2%
}

\renewcommand{\vec}[1]{\bm{#1}}
\newcommand{\mat}[1]{\mathbf{#1}}

\DeclarePairedDelimiterX{\card}[1]{\lvert}{\rvert}{#1}
\DeclarePairedDelimiterX{\tuple}[1]{\lparen}{\rparen}{#1}
\DeclarePairedDelimiterX{\parens}[1]{\lparen}{\rparen}{#1}
\DeclarePairedDelimiterX{\brackets}[1]{\lbrack}{\rbrack}{#1}
\DeclarePairedDelimiterX{\set}[1]\{\}{#1}
\let\Pr\relax
\DeclarePairedDelimiterXPP{\Pr}[1]{\mathbb{P}}[]{}{#1}
\DeclarePairedDelimiterXPP{\PrX}[2]{\mathbb{P}_{#1}}[]{}{#2}
\DeclarePairedDelimiterXPP{\Ex}[1]{\mathbb{E}}[]{}{#1}
\DeclarePairedDelimiterXPP{\ExX}[2]{\mathbb{E}_{#1}}[]{}{#2}

\usepackage{cleveref}

\usetikzlibrary{fit,shapes.misc,arrows.meta}
\makeatletter
\tikzset{
  fitting node/.style={
    inner sep=0pt,
    fill=none,
    draw=none,
    reset transform,
    fit={(\pgf@pathminx,\pgf@pathminy) (\pgf@pathmaxx,\pgf@pathmaxy)}
  },
  reset transform/.code={\pgftransformreset}
}
\tikzset{cross/.style={path picture={
  \draw[black]
(path picture bounding box.south east) -- (path picture bounding box.north west) (path picture bounding box.south west) -- (path picture bounding box.north east);
}}}
\tikzstyle{ox}=[semithick,draw=black,circle,cross,inner sep=1.2mm]
\makeatother

\newcommand{\nc}{\newcommand}
\newcount\Comments
\nc\io[1]{\ifnum\Comments=1 {\textcolor{purple}{[ioannis: #1]}}\fi}
\nc\br[1]{\ifnum\Comments=1 {\textcolor{teal}{[brian: #1]}}\fi}
\nc\et[1]{\ifnum\Comments=1 {\textcolor{orange}{[ET: #1]}}\fi}

\nc{\Opthedge}{OMWU\xspace}

\nc{\DMO}{\DeclareMathOperator}
\nc\old[1]{\textcolor{brown}{[old: #1]}}
\nc{\BR}{\mathbb{R}}
\nc{\BC}{\mathbb{C}}
\DMO{\Bin}{Bin}
\nc{\BN}{\mathbb{N}}
\nc{\distrs}[1]{\Delta({#1})}
\nc{\BZ}{\mathbb{Z}}
\nc{\ep}{\epsilon}
\nc{\ra}{\rightarrow}
\nc{\st}{\star}
\nc{\Reg}[2]{\REG_{{#1},{#2}}}
\nc{\til}{\tilde}
\nc{\kld}[2]{\KL({#1};{#2})}
\nc{\chisq}[2]{\chi^2({#1};{#2})}
\DMO{\POLYLOG}{polylog}

\nc{\matx}[1]{\left(\begin{matrix}#1\end{matrix}\right)}
\DMO{\VAR}{Var}
\DMO{\COV}{Cov}
\nc{\Var}[2]{\VAR_{{#1}}\left({#2}\right)}
\nc{\Cov}[3]{\COV_{{#1}}\left({#2},{#3}\right)}
\DMO{\DD}{D}
\nc{\fd}[2]{\DD_{#1}{#2}}
\nc{\fds}[3]{\left(\fd{#1}{#2}\right)\^{#3}}
\nc{\fdc}[2]{\DD^\circ_{#1}{#2}}
\nc{\fdcs}[3]{\left(\fdc{#1}{#2}\right)\^{#3}}
\nc{\shf}[2]{\EEE_{#1}{#2}}
\nc{\shfs}[3]{\left(\shf{#1}{#2}\right)\^{#3}}
\nc{\normst}[2]{\left\| {#2} \right\|_{#1}^\st}
\renewcommand{\^}[1]{^{(#1)}}
\DeclareMathOperator*{\argmax}{arg\,max}

\nc{\lng}{\langle}
\nc{\rng}{\rangle}
\nc{\bbone}{\mathbf{1}}
\nc{\bbzero}{\mathbf{0}}
\nc{\MD}{\mathcal{D}}
\nc{\MM}{\mathcal{M}}
\nc{\MZ}{\mathcal{Z}}
\nc{\MU}{\mathcal{U}}
\nc{\MR}{\mathcal{R}}
\nc{\MC}{\mathcal{C}}
\nc{\MT}{\mathbb{T}^{n}}
\nc{\MS}{\mathcal{S}}
\nc{\MX}{\mathcal{X}}
\nc{\MY}{\mathcal{Y}}
\nc{\MB}{\mathcal{B}}
\nc{\MJ}{\mathcal{J}}
\nc{\MF}{\mathcal{F}}
\nc{\MG}{\mathcal{G}}
\nc{\ML}{\mathcal{L}}
\nc{\MQ}{\mathcal{Q}}

\nc{\ba}{\mathbf{A}}
\nc{\bx}{\mathbf{x}}
\nc{\by}{\mathbf{y}}
\nc{\bz}{\mathbf{z}}
\nc{\bs}{\mathbf{s}}
\nc{\bt}{\mathbf{t}}

\nc{\ME}{\mathcal{E}}
\DMO{\View}{View}
\DMO{\KL}{KL}
\nc{\MW}{\mathcal{W}}
\nc{\CS}{\mathscr{S}}
\nc{\CI}{\mathscr{I}}
\nc{\CQ}{\mathscr{Q}}
\nc{\CL}{\mathscr{L}}
\nc{\CM}{\mathscr{M}}
\nc{\CG}{\mathscr{G}}
\nc{\CR}{\mathscr{R}}

\nc{\wh}{\widehat}

\nc{\BM}{BM\xspace}
\nc{\ALG}{\texttt{ALG}}
\nc{\MCT}{{\rm MCT}}

\nc{\matrixLL}{\mat{L}}
\nc{\vectorecks}{\vec{x}}
\nc{\vectorLL}{\vec{\ell}}
\nc{\matrixKYU}{\mat{Q}}
\nc{\rowdot}{\cdot}

\nc{\X}{\mathcal{X}}
\nc{\Y}{\mathcal{Y}}

\newcommand{\delimit}[3]{\newcommand{#1}[1]{\left#2##1\right#3}}
\delimit \ceil \lceil \rceil
\delimit \floor \lfloor \rfloor

\definecolor{p0color}{RGB}{0,0,0}
\definecolor{p1color}{RGB}{31,119,180}
\definecolor{p2color}{RGB}{214,39,40}

\def\va{{\bm{a}}}

\def\vg{{\bm{g}}}

\def\vr{{\bm{r}}}

\def\vu{{\bm{u}}}

\def\vx{{\bm{x}}}

\renewcommand\vec\bm

\title{Convergence of Regret Matching in Potential Games and Constrained Optimization}

\author[1]{Ioannis Anagnostides}
\author[1,2]{Emanuel Tewolde}
\author[3]{Brian Hu Zhang}
\author[4]{Ioannis Panageas}
\author[1,2,5]{Vincent Conitzer}
\author[1,6]{Tuomas Sandholm}

\affil[1]{Carnegie Mellon University}
\affil[2]{Foundations of Cooperative AI Lab (FOCAL)}
\affil[4]{University of California, Irvine}
\affil[3]{Massachusetts Institute of Technology}
\affil[5]{University of Oxford}
\affil[6]{Additional affiliations: Strategy Robot, Inc., Strategic Machine, Inc., Optimized Markets, Inc.}
\affil[ ]{}
\affil[ ]{\texttt{\{ianagnos,etewolde,conitzer,sandholm\}}\texttt{@cs.cmu.edu}, \texttt{zhangbh}\texttt{@csail.mit.edu},\texttt{ipanagea@ics.uci.edu}}

\begin{document}

\maketitle

\begin{abstract}
    \textit{Regret matching ($\RM$)}---and its modern variants---is a foundational online algorithm that has been at the heart of many AI breakthrough results in solving benchmark zero-sum games, such as poker. Yet, surprisingly little is known so far in theory about its convergence beyond two-player zero-sum games. For example, whether regret matching converges to Nash equilibria in \textit{potential games} has been an open problem for two decades. Even beyond games, one could try to use $\RM$ variants for general constrained optimization problems. Recent empirical evidence suggests that they---particularly \emph{regret matching$^+$ ($\RMplus$)}---attain strong performance on benchmark constrained optimization problems, outperforming traditional gradient descent-type algorithms.

    We show that $\RMplus$ converges to an $\epsilon$-KKT point after $O_\epsilon(1/\epsilon^4)$ iterations, establishing for the first time that it is a sound and fast first-order optimizer. Our argument relates the KKT gap to the accumulated \emph{regret}, two quantities that are entirely disparate in general but interact in an intriguing way in our setting, so much so that when regrets are bounded, our complexity bound improves all the way to $O_\epsilon(1/\epsilon^2)$. From a technical standpoint, while $\RMplus$ does \emph{not} have the usual one-step improvement property in general, we show that it does in a certain region that the algorithm will quickly reach and remain in thereafter. In sharp contrast, our second main result establishes a lower bound: $\RM$, with or without alternation, can take an exponential number of iterations to reach a crude approximate solution even in two-player potential games. This represents the first worst-case separation between $\RM$ and $\RMplus$. Our lower bound shows that convergence to coarse correlated equilibria in potential games is exponentially faster than convergence to Nash equilibria.
\end{abstract}

\section{Introduction}

\emph{Regret matching} is a foundational online algorithm for minimizing \emph{regret}. It was famously introduced by~\citet{Hart00:Simple}, although its conception can be traced much further back to the seminal \emph{approachability} framework of~\citet{Blackwell56:analog}, which lay the groundwork for online learning and regret minimization. As the name suggests, regret matching prescribes playing each action with probability proportional to the (nonnegative) regret accumulated by that action. Its appeal lies in its simplicity and scalability, being both \emph{parameter free} and \emph{scale invariant}.

Regret matching---and modern versions thereof---has been at the forefront of equilibrium computation in massive two-player zero-sum games. A notable variant with strong empirical performance is \emph{regret matching$^+$}, introduced by~\citet{Tammelin14:Solving}; the only difference is that it truncates all negative coordinates of the regret vector to zero in every iteration. Even so, this variant is typically far superior than its predecessor, and was a central component in AI poker breakthroughs~\citep{Bowling15:Heads,Brown17:Superhuman,Brown19:Superhuman,Moravvcik17:DeepStack} and a more recent superhuman agent for dark chess~\citep{Zhang25:General}.

As such, the regret matching family of algorithms has rightfully been the subject of intense study in contemporary research. Much of this focus has been confined to two-player zero-sum games, where minimizing regret translates to convergence---of the \emph{average} strategies---to minimax (equivalently, Nash) equilibria~\citep{Freund99:Adaptive}. More broadly, in general-sum games, no-regret algorithms guarantee convergence to the set of \emph{coarse correlated equilibria}~\citep{Moulin78:Strategically}---a more permissive concept than Nash equilibria.

In this paper, we examine the convergence of regret matching and its variants in the seminal class of \emph{potential games}, and, more broadly, nonconvex optimization constrained over a product of simplices. Surprisingly little is known about this question even though it was identified early on as an important open question in this space~\citep{Kleinberg09:Multiplicative,Marden07:Regret}. Recent empirical evidence brings this question to the fore again: \citet{Tewolde25:Decision} showed that the regret matching family---and especially regret matching$^+$---attains strong performance on a benchmark suite of constrained optimization problems, significantly outperforming gradient descent-type algorithms. Yet, there is no theory to suggest that regret matching will even asymptotically converge to approximate KKT points in constrained optimization, which are tantamount to Nash equilibria when dealing specifically with potential games. We fill this gap in this paper.

\subsection{Our results}

We analyze the convergence of regret matching $(\RM)$ and regret matching $(\RMplus)$ in the general class of (nonconvex) optimization problems constrained over a product of probability simplices. This encompasses as a special case Nash equilibria in potential games when the objective is multilinear. More broadly, to have a unifying treatment of both settings, we think of each probability simplex as being controlled by a single player who is observing the corresponding part of the gradient.

We cover both the simultaneous and the alternating version of $\RMplus$---whereby players update their strategies one after the other, akin to coordinate descent. Our result for $\RMplus$ is summarized below.

\begin{theorem}
    \label{theorem:main-RMplus}
    $\RMplus$ converges to an $\epsilon$-KKT point of any optimization problem over a product of simplices after $O_\epsilon(1/\epsilon^4)$ iterations.
\end{theorem}
This theorem confirms that $\RMplus$ is a sound and efficient first-order optimizer, lending further credence to the empirical results of~\citet{Tewolde25:Decision}. We hope that~\Cref{theorem:main-RMplus} will help cement $\RMplus$ in the optimization arsenal going forward.

We remark that for potential games and constrained optimization over a single simplex, the $O_\epsilon(1/\epsilon^4)$ bound holds no matter how the regrets in (alternating) $\RMplus$ are initialized. For constrained optimization over multiple simplices---with or without alternation---we obtain the same rate by suitably initializing the regret vectors (\Cref{cor:easycor,cor:easycor-sim}); for the usual parameter-free version of $\RMplus$, in which the regret vectors are initialized at zero, we can only guarantee an inferior bound growing as $O_\epsilon(1/\epsilon^8)$ (\Cref{theorem:badrate}).

Our argument proceeds by parameterizing the rate of convergence of $\RMplus$ as a function of the accumulated regret, so much so that if the regret with respect to each individual simplex remains bounded, the rate is improved all the way to $T^{-1/2}$.

\begin{theorem}
    \label{theorem:main-refinedRM+}
    Suppose that the regret of $\RMplus$ on each individual simplex grows as at most $T^{\alpha}$ for some $\alpha \in [0, \nicefrac{1}{2} ]$. Then $\RMplus$ converges to an $\epsilon$-KKT point after $O_\epsilon( 1 / \epsilon^{\nicefrac{2}{1-\alpha}})$ iterations. 
\end{theorem}
$\RMplus$ always guarantees regret growing as $\sqrt{T}$, so~\Cref{theorem:main-RMplus} is implied by~\Cref{theorem:main-refinedRM+}. What makes the latter theorem surprising is that, in general, regret is a fundamentally disparate property compared to KKT gap: as we point out in~\Cref{prop:4cycle}, a sequence can incur zero regret while having an $\Omega(1)$ KKT gap in each iteration. Even so, \Cref{theorem:main-refinedRM+} directly relates the KKT gap in terms of the regret. In particular, the non-asymptotic rate of~\Cref{theorem:main-RMplus} is a consequence of the fact that $\RMplus$ has the no-regret property! In the special case of potential games, regret is known to drive the rate of convergence to \emph{coarse correlated equilibria (CCE)}; \Cref{theorem:main-refinedRM+} shows, for the first time, that regret can also govern the rate of convergence to Nash equilibria. In particular, if convergence to CCE happens at a rate of $T^{-(1-\alpha)}$, for some $\alpha \in [0, \nicefrac{1}{2}]$, the rate of convergence to Nash equilibria is no slower than $T^{-\frac{1-\alpha}{2}}$.

From a technical standpoint, the key challenge is that $\RMplus$ does \emph{not} have a one-step improvement property: even if one initializes $\RMplus$ close to a KKT point, $\RMplus$ can still grossly overshoot. And, of course, it is a parameter-free algorithm, so the usual treatment of gradient descent-type algorithms that relies on appropriately tuning the learning rate falls short. In this context, our starting observation is that, at least when the utility function is linear, $\RMplus$ is bound to improve the utility, although the improvement is inversely proportional to the norm of the regret vector (\Cref{lemma:onestepRM+}). This key property already suffices to show that alternating $\RMplus$ will converge to Nash equilibria in potential games. For the more challenging setting where the updates are simultaneous or the objective is not multilinear, we first show that one-step improvement holds conditional on the norm of the regret vector being \emph{sufficiently large} (\Cref{lemma:R-L,lemma:R-L-multi}). To conclude the argument, we combine this property with the crucial insight that the $\ell_2$ norm of the regret vector is \emph{monotonically increasing} proportionally to the KKT gap (\Cref{prop:increasing-regret}). This means that $\RMplus$ will never get stuck in a cycle: the regret vector would quickly grow in norm, at which point the one-step improvement promised by~\Cref{lemma:R-L} kicks in.

Does $\RM$ share the same convergence properties as $\RMplus$? As a reminder, the only difference is that $\RM$ refrains from truncating negative regrets to zero. Even so, we find that this seemingly innocuous difference gives rise to an exponential gap in the performance of $\RM$ \emph{vis-\`a-vis} $\RMplus$, manifested even in two-player identical-interest games---a special case of potential games (\Cref{fig:exp-sep}).

\begin{theorem}
    \label{theorem:main-RMslow}
    There is a two-player $m \times m$ identical-interest game where $\RM$, with or without alternation, requires $m^{\Omega(m)}$ iterations to converge to an $m^{-\Theta(1)}$-approximate Nash equilibrium.
\end{theorem}
This lower bound holds not just under an adversarial initialization, but also when players initially mix uniformly at random, which is the most common initialization in practice.

\Cref{theorem:main-RMslow} constitutes the first worst-case separation---let alone an exponential one---between $\RM$ and $\RMplus$. Indeed, in zero-sum games, it is known that $\RM$ and $\RMplus$ both attain a rate no faster than $T^{-1/2}$~\citep{Farina23:Regret}, even though $\RMplus$ typically performs much better in practice. \Cref{theorem:main-RMslow} provides further justification for opting for $\RMplus$ instead of $\RM$, albeit in a fundamentally different setting.

The basic flaw of $\RM$ that underpins~\Cref{theorem:main-RMslow} is that, even with a linear utility, it is not guaranteed to improve the utility even when it has a large best-response gap; specifically, as we show in~\Cref{lemma:onestepRM}, the improvement is conditional on a good-enough action having nonnegative regret. But herein lies the problem: it could take many iterations before the regret resurfaces to being positive. What happens in the construction behind~\Cref{theorem:main-RMslow} is that it takes longer and longer---exponentially so---for the regret of the unique good-enough action to be positive; before then, $\RM$ is entirely stalled without making any progress. At the same time, $\RM$ is guaranteed to converge to the set of coarse correlated equilibria (CCE) at a rate of $T^{-1/2}$, simply because it always has the no-regret property. This leads to the following interesting consequence.

\begin{corollary}
    There is a class of two-player potential games in which $\RM$ converges to an $\epsilon$-CCE in $O_\epsilon(1/\epsilon^2)$ rounds but it takes $\exp(\Omega(1/\epsilon))$ rounds to converge to an $\epsilon$-Nash equilibrium.
\end{corollary}
To be clear, convergence to a CCE is meant in terms of the average correlated distribution of play, whereas convergence to Nash equilibria is in terms of the individual iterates produced by $\RM$.

\begin{figure}
    \centering
    \includegraphics[scale=0.5]{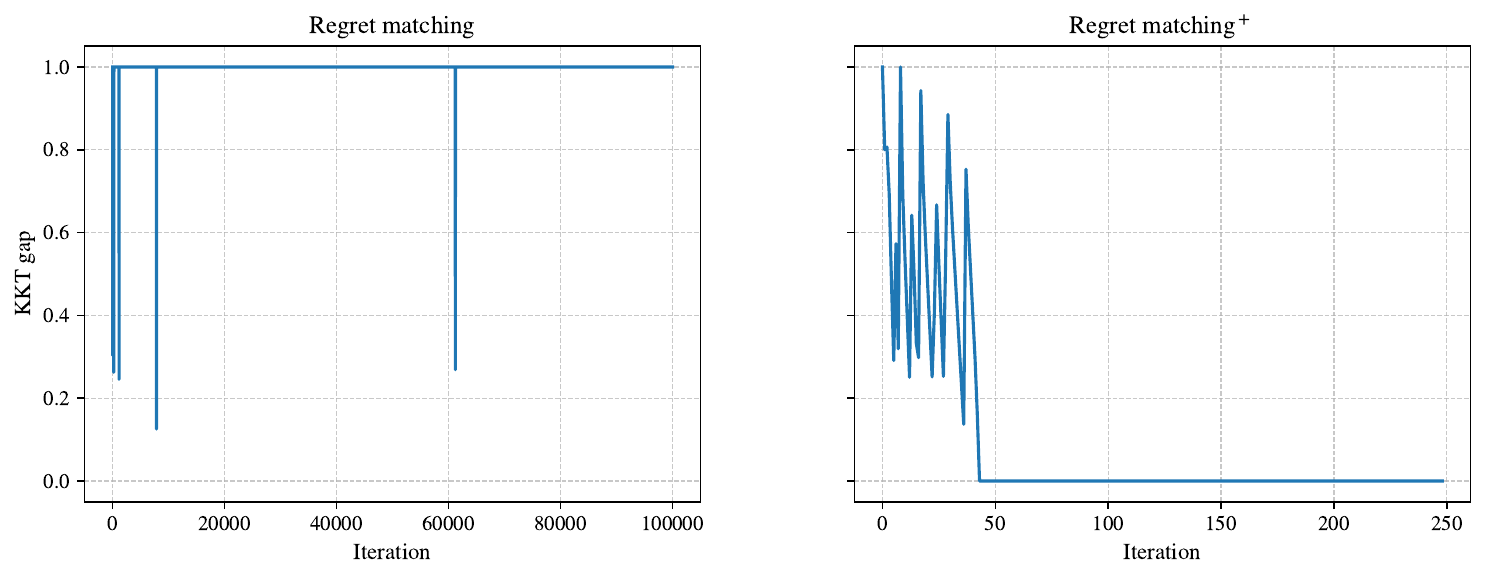}
    \caption{Illustration of our main results: $\RMplus$ always converges fast to a KKT point while $\RM$ can take exponential time even in two-player identical-interest games, constructed in~\Cref{sec:RM}.}
    \label{fig:exp-sep}
\end{figure}

We defer further discussion on related work to~\Cref{sec:related}.

\section{Preliminaries}

We begin by introducing potential games and the more general problem of constrained optimization over a product of simplices (\Cref{sec:potential-prel}), and then recall regret matching($^+$) in~\Cref{sec:onlinelearning}.

\paragraph{Notation} For a vector $\vx \in \R^\cA$, we access its $a$th coordinate by $\vx[a]$. For two vectors $\vx, \vx' \in \R^{\cA}$ of compatible dimension, $\langle \vx, \vx' \rangle \defeq \sum_{a \in \cA} \vx[a] \vx'[a]$ denotes their inner product. For $\vx \in \R^\cA$, we let $\|\vx\|_1 \defeq \sum_{a \in \cA} | \vx[a]|$, $\|\vx\|_2 \defeq \sqrt{ \langle \vx, \vx \rangle  }$, and $\|\vx\|_\infty \defeq \max_{a \in \cA} |\vx[a]|$. We generally use subscripts to indicate the player and superscripts for the (discrete) time index. For simplicity in exposition, we use the notation $O_n(\cdot)$ to suppress the dependence on all parameters except $n$. $O(\cdot), \Omega(\cdot), \Theta(\cdot)$ hide absolute constants.

\subsection{Potential games and constrained optimization}
\label{sec:potential-prel}

\paragraph{Normal-form games} Our first key focus in this paper is on \emph{potential games}, which we represent in the usual normal form. Here, we have $n$ players, each of whom is to select an action $a_i$ from a finite set $\cA_i$, with $m_i \defeq |\cA_i|$ and $m = \max_{1 \leq i \leq n} m_i$. Under a joint action profile $\vec{a} = (a_1, \dots, a_n) \in \cA_1 \times \dots \times \cA_n$, each player $i \in [n]$ receives a payoff given by a \emph{utility function} $u_i: (a_1, \dots, a_n) \mapsto u_i(a_1, \dots, a_n) \in \R$ with range bounded by $1$. A player $i \in [n]$ can randomize by specifying a \emph{mixed strategy} $\vx_i \in \Delta(\cA_i) \defeq \{ \vx_i \in \R_{\geq 0}^{\cA_i} : \sum_{a_i \in \cA_i} \vx_i[a_i] = 1 \} $. Player $i$ strives to maximize its \emph{expected} utility, given by $u_i(\vx_1, \dots, \vx_n) \defeq \sum_{(a_1, \dots, a_n) \in \cA_1 \times \dots \times \cA_n} u_i(a_1, \dots, a_n) \prod_{i' = 1}^n \vx_{i'}[a_{i'}] $. A key fact is that the expected utility is \emph{multilinear}, in that $u_i(\vx_1, \dots, \vx_n) = \langle \vx_i, \vu_i(\vx_{-i})\rangle$ for some utility vector $\vu_i(\vx_{-i}) \in \R^{\cA_i}$ that does not depend on $\vx_i$; namely, $\vu_i(\vx_{-i}) = (\sum_{(a_1, \dots, a_{i-1}, a_{i+1}, \dots, a_n)} u_i(a_1, \dots, a_n) \prod_{i' \neq i} \vx_{i'}[a_{i'}] )_{a_i \in \cA_i}$. Here and throughout, we use the shorthand notation $\vx_{-i} = (\vx_1, \dots, \vx_{i-1}, \vx_{i+1}, \dots, \vx_n)$, while we recall that $\langle \cdot, \cdot \rangle$ denotes the inner product. Further, we use the shorthand notation $\BRgap_i(\vx_i, \vu_i) \defeq \max_{\vx_i' \in \Delta(\cA_i)} \langle \vx_i' - \vx_i, \vu_i \rangle$ for the best-response gap.

The predominant solution concept in game theory is the \emph{Nash equilibrium}~\citep{Nash50:Equilibrium}.

\begin{definition}
    \label{def:NE}
    A strategy profile $(\vx_1, \dots, \vx_n) \in \Delta(\cA_1) \times \dots \times \Delta(\cA_n)$ is an \emph{$\epsilon$-Nash equilibrium} if for any player $i \in [n]$ and unilateral deviation $\vx_i' \in \Delta(\cA_i)$, 
    $$u_i(\vx_i', \vx_{-i}) \leq u_i(\vx_i, \vx_{-i}) + \epsilon.$$
\end{definition}
A standard relaxation of the Nash equilibrium is the \emph{coarse correlated equilibrium} (\Cref{def:CCE}), which can be attained by no-regret algorithms (\Cref{prop:CCE}). While finding a Nash equilibrium is hard even in two-player general-sum games~\citep{Daskalakis08:Complexity,Chen09:Settling}, our focus is on \emph{potential games}---equivalently, \emph{congestion games}~\citep{Monderer96:Potential}.

\paragraph{Potential games} This is a seminal class that goes back to the work of~\citet{Rosenthal73:Class}. The defining property is the admission of a global, player-independent function---the \emph{potential}---whose difference reflects the benefit of any unilateral deviation.

\begin{definition}[Potential game]
    \label{def:potential}
    An $n$-player game is a \emph{potential game} if there exists a function $\Phi : \cA_1 \times \dots \times \cA_n \to \R$ such that for any player $i \in [n]$ and actions $a_i, a_i' \in \cA_i$, $\vec{a}_{-i} \in \bigtimes_{i' \neq i} \cA_{i'}$,
    \begin{equation}
        \label{eq:pot-property}
        \Phi(a_i', \vec{a}_{-i}) - \Phi(a_i, \vec{a}_{-i}) = u_i(a_i', \vec{a}_{-i}) - u_i(a_i, \vec{a}_{-i}).
    \end{equation}
\end{definition}
With an overload in notation, we also denote by $\Phi : \Delta(\cA_1) \times \dots \times \Delta(\cA_n)$ the mixed extension of the utility: $\Phi(\vx_1, \dots, \vx_n) = \E_{(a_1, \dots, a_n) \sim (\vx_1, \dots, \vx_n)} \Phi(a_1, \dots, a_n)$. By~\eqref{eq:pot-property}, it is easy to see that $\Phi(\vx_i', \vx_{-i}) - \Phi(\vx_i, \vx_{-i}) = u_i(\vx_i', \vx_{-i}) - u_i(\vx_i, \vx_{-i})$ for any strategies $\vx_i, \vx_{i}' \in \Delta(\cA_i)$ and $\vx_{-i} \in \bigtimes_{i' \neq i} \Delta(\cA_{i'})$.

A special case of a potential game worth noting is an \emph{identical-interest} game, which means that $u_1(\vx_1, \dots, \vx_n) = \dots = u_n(\vx_1, \dots, \vx_n)$ for all $\vx_1, \dots, \vx_n$. In the presence of only two players, this simplifies to $u_1(\vx_1, \vx_2) = \langle \vx_1, \mat{A} \vx_2 \rangle = u_2(\vx_1, \vx_2)$ for a common payoff matrix $\mat{A} \in \R^{\cA_1 \times \cA_2}$; for the sake of brevity, this will sometimes be referred to as an $m_1 \times m_2$ game, where $m_1 = |\cA_1|$ and $m_2 = |\cA_2|$. In a potential game, it holds that $\vu_i(\vx_{-i}) = \nabla_{\vx_i} \Phi(\vx)$ for any player $i \in [n]$.

A (mixed) Nash equilibrium in potential games is amenable to (projected) gradient descent, but is likely hard to compute when the precision $\epsilon > 0$ is exponentially small~\citep{Babichenko21:Settling}. Our focus will be on algorithms whose complexity is polynomial in $1/\epsilon$.

\paragraph{Constrained optimization} More broadly, beyond potential games, we are interested in computing \emph{Karush-Kuhn-Tucker (KKT) points} of a function $u : \cX \to \R$, where $\cX \defeq \Delta(\cA_1) \times \dots \times \Delta(\cA_n)$. We assume that $u$, which is to be maximized, is differentiable over an open set $\hat{\cX} \supset \cX$ and \emph{$L$-smooth}, meaning that $\| \nabla u(\vx) - \nabla u(\vx') \|_2 \leq L\|\vx - \vx'\|_2$ for all $\vx, \vx' \in \cX$; we recall that $\|\vx \|_2 \defeq \sqrt{ \langle \vx, \vx \rangle } $ denotes the (Euclidean) $\ell_2$ norm. We make the normalization assumption $| \langle \vx_i - \vx_i', \nabla_{\vx_i} u(\vx) \rangle | \leq 1$ for all $i \in [n]$ and $\vx_i, \vx_i' \in \Delta(\cA_i)$. The goal is to minimize the \emph{KKT gap}, which we measure by\footnote{There are other ways of measuring the KKT gap (\emph{e.g.},~\citealp{Fearnley22:Complexity}), which are equivalent in our setting.}
\begin{equation}
    \label{eq:KKT}
    \KKTgap: \cX \ni \vx \mapsto \max_{\vx' \in \cX} \langle \vx' - \vx, \nabla u(\vx) \rangle = \sum_{i=1}^n \BRgap_i(\vx_i, \nabla_{\vx_i} u(\vx)) .
\end{equation}
A point with small KKT gap per~\eqref{eq:KKT} is also referred to as an approximate \emph{first-order stationary point}, which is an approximate fixed point of the (constrained) gradient descent mapping $\vx \mapsto \proj_{\cX}( \vx + \eta \nabla u(\vx))$, where $\eta \leq \nicefrac{1}{L}$ and $\Pi_\cX(\cdot)$ is (Euclidean) projection mapping. A potential game can be seen as the special case in which $u$ is multilinear.\footnote{We caution that if we use the KKT gap per~\eqref{eq:KKT} in the special case of potential games we get the \emph{sum} of the players' deviation benefits, while the approximation in the Nash equilibrium is defined with respect to the \emph{max}.} 

One class of problems that fits in this framework of constrained optimization over a product of simplices is \emph{imperfect-recall} games; we point the interested reader to~\citet{Gimbert20:Bridge,Waugh09:Practical,Koller92:Complexity,Piccione97:Interpretation,Tewolde25:Decision,Berker25:Value} and the references therein. Our results imply that, in imperfect-recall decision problems, $\RMplus$ converges to what is referred to as \emph{CDT equilibria}, which are tantamount to KKT points of the utility function~\citep{Tewolde23:Computational}.

\subsection{Online learning and regret matching}
\label{sec:onlinelearning}

Moving on, we now introduce $\RM$ and $\RMplus$ within the framework of online learning. Here, a \emph{learner} interacts with an \emph{environment} over a sequence of $T$ rounds. In each round $t \in [T]$, the learner first elects a mixed strategy $\vx \in \Delta(\cA)$. The environment in turn specifies a linear utility function $u^{(t)} : \vx \mapsto \langle \vx, \vu^{(t)} \rangle$ for some utility vector $\vu^{(t)} \in \R^{\cA}$; it is assumed that $u^{(t)}$ has a range bounded by $1$. In the full-feedback setting, $\vu^{(t)}$ is revealed to the learner at the end of the $t$th round. The performance of the learner in this online environment is evaluated through \emph{regret},
\begin{equation}
    \label{eq:reg}
    \reg^{(T)} \defeq \max_{\vx' \in \Delta(\cA)} \sum_{t=1}^T \langle \vx' - \vx^{(t)}, \vu^{(t)} \rangle.
\end{equation}
Two popular algorithms for minimizing regret on the simplex are regret matching $(\RM)$ and regret matching$^+$ $(\RMplus)$, formally defined in~\Cref{alg:regret_matching,alg:regret_matching_plus}. They both prescribe playing an action with probability proportional to the nonnegative regret accumulated by that action. Their \emph{only} difference is that $\RMplus$ always truncates the regret to $0$ (\Cref{line:regupdate+}); in that line, $\vec{1}$ denotes the all-ones vector, whose dimension is omitted as it is clear from the context, and $[\cdot]^+ \defeq \max(\vec{0}, \cdot)$ is the nonnegative part.

\begin{proposition}[\citealp{Zinkevich07:Regret,Farina21:Faster}]
    \label{prop:noregret-RM}
    For any sequence of utilities $(\vu^{(t)})_{t=1}^T$, both $\RM$ and $\RMplus$ guarantee that the $\ell_2$ norm of $[\vr^{(T)}]^+$ is at most $\sqrt{m T}$.
\end{proposition}
In particular, for both $\RM$ and $\RMplus$, $\reg^{(T)} \leq \| [\vr^{(T)}]^+ \|_\infty \leq \| [\vr^{(T)}]^+ \|_2 \leq \sqrt{m T}$.

\begin{minipage}[t]{0.48\textwidth}
\begin{algorithm}[H]
\caption{Regret matching ($\RM$)}
\label{alg:regret_matching}
Initialize cumulative regrets $\vR^{(0)} \leftarrow \vec{0}$\;
Initialize strategy $\vx^{(0)} \in \Delta(\cA)$\;
\For{$t = 1, \dots, T$}{
    Set $\vtheta^{(t)} \leftarrow [\vR^{(t-1)} ]^+$\;
    \If{$ \vtheta^{(t)} \neq \vec{0}$} {
    Compute $\vx^{(t)} \leftarrow \nicefrac{\vtheta^{(t)}}{\| \vtheta^{(t)} \|_1 } $\;
    }
    \Else{
        $\vx^{(t)} \leftarrow \vx^{(t-1)}$\;
    }
    Output strategy $\vx^{(t)} \in \Delta(\cA)$ \;
    Observe utility $\vu^{(t)} \in \R^{\cA}$\;
    $\vR^{(t)} \leftarrow \vR^{(t-1)} + \vu^{(t)} - \langle \vx^{(t)}, \vu^{(t)} \rangle \vec{1} $\;
}
\end{algorithm}
\end{minipage}
\hfill
\begin{minipage}[t]{0.48\textwidth}
\begin{algorithm}[H]
\caption{Regret matching$^+$ $(\RMplus)$}
\label{alg:regret_matching_plus}
\SetKwInOut{Input}{Input}
\SetKwInOut{Output}{Output}
Initialize cumulative regrets $\vR^{(0)} \defeq \vec{0}$\;
Initialize strategy $\vx^{(1)} \in \Delta(\cA)$\;
\For{$t = 1, \dots, T$}{
    Set $\vtheta^{(t)} \leftarrow \vR^{(t-1)}$\;
    \If{$ \vtheta^{(t)} \neq \vec{0}$} {
    Compute $\vx^{(t)} \leftarrow \nicefrac{\vtheta^{(t)}}{\| \vtheta^{(t)} \|_1 } $\;
    }
    \Else{
        $\vx^{(t)} \leftarrow \vx^{(t-1)}$\;
    }
    Output strategy $\vx^{(t)} \in \Delta(\cA)$ \;
    Observe utility $\vu^{(t)} \in \R^{\cA}$\;
    $\vR^{(t)} \leftarrow [\vR^{(t-1)} + \vu^{(t)} - \langle \vx^{(t)}, \vu^{(t)} \rangle \vec{1}]^+$\;\label{line:regupdate+}
}
\end{algorithm}
\end{minipage}

\paragraph{Simultaneous and alternating updates}

We are interested in the convergence of $\RM$ and $\RMplus$ when used by all players; in the constrained optimization setting, we think of having one player acting on each simplex, in direct correspondence with potential games. In this setting, the sequence of utilities $(\vu_i^{(t)})_{t=1}^T$ given as input to player $i \in [n]$ is determined by the strategies of the other players. If the updates are \emph{simultaneous}, we have $\vu_i^{(t)} = \nabla_{\vx_i} u(\vx^{(t)})$ for each player $i \in [n]$. (In potential games, the potential function $\Phi$ plays the role of $u$, so that $\vu_i^{(t)} = \vu_i(\vx_{-i}^{(t)})$.) In the alternating setting, we go through the players in a round-robin fashion $i = 1, \dots, n$. In the \emph{lazy} version of the update, for a fixed precision $\epsilon > 0$, we first compute $\vu_i^{(t)} = \nabla_{\vx_i} u( \vx_{i' < i}^{(t+1)}, \vx_{i' \geq i}^{(t)} )$. If the best-response gap of player $i \in [n]$ is already at most $\epsilon$, we refrain from updating that player, so that $\vx_i^{(t+1)} \defeq \vx_i^{(t)}$. Otherwise, the player updates its strategy to $\vx_i^{(t+1)}$ using $\vu_i^{(t)}$. We refer to this scheme as \emph{$\epsilon$-lazy alternating} updates ($\epsilon$-lazy simultaneous updates are defined similarly); one limitation of this lazy variant is that it is not an ``anytime algorithm,'' in that one needs to specify the precision beforehand. In the more common, non-lazy version of alternation, a player is updated regardless of the best-response gap; in the sequel, we obtain convergence results for both variants. Alternation speeds up performance, at least in zero-sum games~\citep{Tammelin14:Solving}, and has been the subject of much recent research~\citep{Wibisono22:Alternating,Cevher23:Alternation,Nan25:Convergence}.

\section{Convergence of regret matching$^+$}
\label{sec:RMplus}

In this section, we analyze the convergence of $\RMplus$ in potential games (\Cref{sec:potential}), and more broadly, constrained optimization (\Cref{sec:nonlinear}). A central theme in our analysis of $\RM$ and $\RMplus$ is a recurring connection between regret and convergence to KKT points.

Before we proceed, it is worth highlighting that, in general, the no-regret property is fundamentally different from convergence to KKT points in \emph{nonconvex problems}. To begin with, we point out that when the underlying function to be maximized, $u$, is concave, then the no-regret property does imply convergence to a global optimum, from Jensen's inequality.

\begin{proposition}[Under concavity, no-regret implies convergence]
    \label{prop:concave}
    Let $u$ be a smooth concave function. If an online algorithm observes the sequence of utilities $( \nabla u(\vx^{(t)}))_{t=1}^T$, then $ \frac{1}{T} \sum_{t=1}^T u(\vx^{(t)}) \geq \max_{\vx \in \cX} u(\vx) - \frac{1}{T} \reg^{(T)}$, where $\reg^{(T)}$ is the regret of the algorithm per~\eqref{eq:reg}.
\end{proposition}
Thus, if the algorithm has vanishing average regret, $u(\vx^{(t)})$ converges to $\max_{\vx \in \cX} u(\vx)$ (in density). But beyond concave problems, no-regret algorithms do not necessarily guarantee convergence even to a KKT point, as we point out below.

\begin{proposition}
    \label{prop:4cycle}
    For any $T \in \N$ with $T = 0 \mod 4$, there exists a polynomial function $u$ in $[0,1]$ and a sequence of points $(x^{(t)})_{t=1}^T$ such that
    \begin{itemize}[noitemsep]
        \item the regret of the sequence with respect to $(\nabla u(x^{(t)}) )_{t=1}^T$ is zero, while
        \item every point in the sequence has an $\Omega(1)$ KKT gap with respect to the function $u$.
    \end{itemize}
\end{proposition}

This is based on the $4$-cycle $0.6 \rightarrow 0.7 \rightarrow 0.4 \rightarrow 0.3 \rightarrow 0.6$. If the gradients observed at those points are $0.6 \mapsto 2, 0.7 \mapsto -1, 0.4 \mapsto -2, 0.3 \mapsto 1$, it follows that i) $\sum_{t=1}^T \nabla u(x^{(t)}) = 0$ and ii) $\sum_{t=1}^T x^{(t)} \nabla u(x^{(t)}) = 0$, which in turn implies that this sequence incurs zero regret. But, by construction, the gradients at those interior points have a large magnitude, which in turn implies that the KKT gap is large. (That the average is a local minimum is coincidental.) A polynomial consistent with the above gradients is $90x-298.\bar{3} x^{2}+416.\bar{6}x^{3}-208.\bar{3}x^{4}$, leading to~\Cref{prop:4cycle}; we note that the above sequence of iterates is not realizable through an algorithm such as gradient descent.

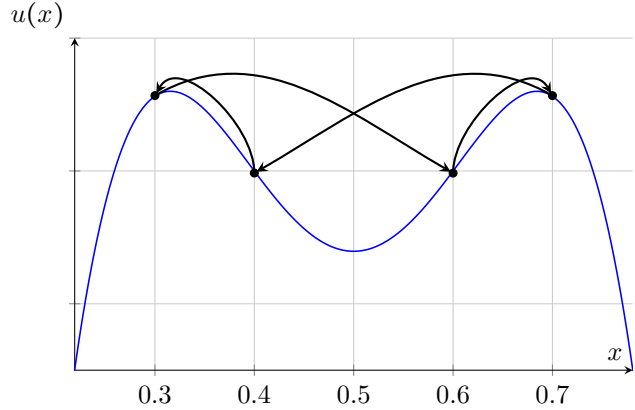
\begin{figure}
    \centering
    \begin{tikzpicture}
  \begin{axis}[
    width=9cm, height=6cm,
    domain=0.2:0.8, samples=1000,
    axis lines=middle,
    xlabel={$x$}, 
    xlabel style={font=\small},   
    xticklabel style={font=\small},  
    ylabel style={font=\small},
    ylabel={$u(x)$},
    ylabel style={at={(axis description cs:0,1)}, anchor=south east},
    grid=both, minor grid style={gray!20}, major grid style={gray!40},
    ymin=9.3, ymax=9.8,
    yticklabels={},
    xtick={0.2,0.3,0.4,0.5,0.6,0.7,0.8}
  ]

    \addplot[semithick, blue, no markers] 
      {90*x - 298.334*x^2 + 416.667*x^3 - 208.333*x^4};

    \pgfmathsetmacro{\yone}{90*0.3 - 298.334*0.3^2 + 416.667*0.3^3 - 208.333*0.3^4}
    \pgfmathsetmacro{\ytwo}{90*0.4 - 298.334*0.4^2 + 416.667*0.4^3 - 208.333*0.4^4}
    \pgfmathsetmacro{\ythree}{90*0.6 - 298.334*0.6^2 + 416.667*0.6^3 - 208.333*0.6^4}
    \pgfmathsetmacro{\yfour}{90*0.7 - 298.334*0.7^2 + 416.667*0.7^3 - 208.333*0.7^4}

    \addplot[only marks, mark size=1.5pt, mark=*] coordinates {(0.3,\yone)} node[coordinate] (P1) {};
    \addplot[only marks, mark size=1.5pt, mark=*] coordinates {(0.4,\ytwo)} node[coordinate] (P2) {};
    \addplot[only marks, mark size=1.5pt, mark=*] coordinates {(0.6,\ythree)} node[coordinate] (P3) {};
    \addplot[only marks, mark size=1.5pt, mark=*] coordinates {(0.7,\yfour)} node[coordinate] (P4) {};

    \draw[->, thick, >=stealth, shorten >=1pt] (P3) to[out=90,in=120] (P4);
    \draw[->, thick, >=stealth, shorten >=1pt] (P4) to[out=150,in=30] (P2);
    \draw[->, thick, >=stealth, shorten >=1pt] (P2) to[out=90,in=60] (P1);
    \draw[->, thick, >=stealth, shorten >=1pt] (P1) to[out=30,in=150] (P3);

  \end{axis}
\end{tikzpicture}
    \caption{The example corresponding to~\Cref{prop:4cycle}, demonstrating that having zero regret, let alone sublinear, has no implications concerning convergence in terms of KKT gap.}
    \label{fig:4cycle}
\end{figure}

\subsection{Potential games}
\label{sec:potential}

We first analyze convergence in potential games. A key property, which paves the way for~\Cref{theorem:altern-RM+}, is that, for a fixed utility vector, $\RMplus$ has a one-step improvement property; the lemma below takes the perspective of a single, arbitrary player in the game.

\begin{restatable}[One-step improvement for $\RMplus$]{lemma}{onestepRMplus}
    \label{lemma:onestepRM+}
    For any $\vR \in \R^\cA_{\geq 0}$ and $\vu \in \R^\cA$, we define $\vx \defeq \nicefrac{\vR}{\| \vR \|_1}$; if $\vR = \vec{0}$, $\vx \in \Delta(\cA)$ can be arbitrary. If $\vR' \defeq [ \vR + \vu - \langle \vx, \vu \rangle \vec{1}]^+ \neq \vec{0}$ and $\vx' \defeq \nicefrac{\vR'}{\| \vR' \|_1}$, then
    \begin{equation}
        \label{eq:improve-RM+}
        \langle \vx' - \vx, \vu \rangle \geq \frac{1}{\| \vR' \|_1} \left( \max_{a \in \cA} \vu[a] - \langle \vx, \vu \rangle \right)^2 = \frac{1}{ \|\vr'\|_1}  \BRgap(\vx, \vu)^2.
    \end{equation}
    If $\vR' = \vec{0}$, then $\langle \vx, \vu \rangle = \langle \vx', \vu \rangle \geq \max_{a \in \cA} \vu[a]$.
\end{restatable}

The left-hand side of~\eqref{eq:improve-RM+} reflects the improvement in utility obtained by updating $\vx$ to $\vx'$; in particular, $\max_{a \in \cA} \vu[a] - \langle \vx, \vu \rangle$ is the best-response gap of $\vx$ with respect to $\vu$. \Cref{lemma:onestepRM+} implies that the utility is monotonically increasing---unless the current strategy is already a best response to $\vu$. Furthermore, so long as the regret vector is \emph{small enough}, the improvement is bound to be substantial, being proportional to the squared best-response gap. It is worth noting that~\Cref{lemma:onestepRM+} holds no matter the initial regret vector $\vR$, subject to $\vR \in \R_{\geq 0}$; this invariance always holds for $\RMplus$ (by definition of the algorithm in~\Cref{line:regupdate+}), but that is not so for $\RM$ (\emph{cf.}~\Cref{lemma:onestepRM}). 

The proof of~\Cref{lemma:onestepRM+} proceeds by expressing~\eqref{eq:improve-RM+} in terms of the regret vectors and performing simple algebraic manipulations; it appears in~\Cref{appendix:proofs-RMplus}. Furthermore, as we point out in~\Cref{lemma:lowerRM+}, \Cref{lemma:onestepRM+} is in a certain sense tight.

\paragraph{Convergence in potential games} We now employ \Cref{lemma:onestepRM+} to show that alternating $\RMplus$ quickly converges to approximate Nash equilibria in potential games. Using the fact that the game admits a potential (per~\Cref{def:potential}), we have that for any round $t \in [T]$, $\Phi(\vx_1^{(t+1)}, \dots, \vx^{(t+1)}_{n} ) - \Phi(\vx_1^{(t)}, \dots, \vx^{(t)}_{n}) \geq \sum_{i=1}^n \frac{1}{\| \vR^{(t)}_i \|_1} \BRgap_i(\vx_i^{(t)}, \vu_i^{(t)})^2 \mathbbm{1} \{ \BRgap_i(\vx_i^{(t)}, \vu_i^{(t)}) > \epsilon \}$, where we used~\Cref{lemma:onestepRM+} together with the assumption that only players with more than $\epsilon$ best-response gap update their strategies. The telescopic summation over $t = 1, \dots, T$ yields
\begin{equation}
    \label{eq:telescopic}
    \Phimax \geq \sum_{t=1}^T \sum_{i=1}^n \frac{1}{ \| \vr_i^{(t)} \|_1 } \BRgap_i(\vx_i^{(t)}, \vu_i^{(t)})^2 \mathbbm{1} \{ \BRgap_i(\vx_i^{(t)}, \vu_i^{(t)} ) > \epsilon \},
\end{equation}
where $\Phimax$ denotes the range of the potential function. If in every round $t \in [T]$ there is a player $i \in [n]$ such that $\BRgap_i(\vx_i^{(t)}, \vu_i^{(t)}) > \epsilon$, we have $\Phimax \geq \sum_{t=1}^T \frac{1}{m \sqrt{t}} \epsilon^2 \geq \frac{1}{m} \epsilon^2 \sqrt{T}$, where we used that $\|\vR_i^{(t)} \|_1 \leq \sqrt{m} \|\vR_i^{(t)} \|_2 \leq m \sqrt{T}$ (\Cref{prop:noregret-RM}). We thus arrive at the following result.

\begin{theorem}
    \label{theorem:altern-RM+}
    In any potential game, $\epsilon$-lazy alternating $\RMplus$ requires at most $1 + \frac{ ( m \Phimax )^2 }{\epsilon^4}$ rounds to converge to an $\epsilon$-Nash equilibrium. More broadly, if $\| \vR_i^{(t)} \|_1 \leq C(n, m) t^{\alpha}$ for all $i \in [n]$ and some $\alpha \in [0, \nicefrac{1}{2}]$, it requires $1 + \frac{ ( C(n, m) \Phimax )^\beta }{\epsilon^{2\beta}}$ rounds, where $\beta \defeq \nicefrac{1}{1 - \alpha}$.
\end{theorem}

This provides a convergence rate of $T^{-1/4}$ to Nash equilibria. Notwithstanding~\Cref{prop:4cycle}, an intriguing aspect of~\Cref{theorem:altern-RM+} is that it \emph{connects convergence to Nash equilibria to the regret} of $\RMplus$. In particular, if $\RMplus$ did not have the no-regret property, meaning that $\|\vr_i^{(t)}\|_1 = \Omega(t)$, one could only prove an exponential bound since $\sum_{t=1}^T \nicefrac{1}{t} = \Theta(\log T)$. At the other end of the spectrum, when each player accumulates constant regret, \Cref{theorem:altern-RM+} implies an improved convergence rate of $T^{-1/2}$. It is an open question whether $\RMplus$, with or without alternation, can experience $\Omega(\sqrt{T})$ regret in potential games.

One caveat of lazy alternating $\RMplus$ prescribed by~\Cref{theorem:altern-RM+} is that the desired precision should be known in advance in order to execute the algorithm. We address this limitation in~\Cref{theorem:nonlazy}, where we show that the usual, non-lazy version also converges after $O_\epsilon(1/\epsilon^4)$ rounds, albeit at the cost of introducing an additional dependence in the total number of iterations.

\paragraph{Faster rates using discounting} Next, we refine~\Cref{theorem:altern-RM+} through the use of \emph{discounted} $\RMplus$, which means that the regret vector is multiplied by a discount factor $\alpha^{(t)} \in (0, 1]$ in each round; we spell out $\DRMplus$ in~\Cref{alg:regret_matching_plus-discounting}. This class of algorithms was introduced by~\citet{Brown19:Solving}, who showed that discounting drastically improves empirical performance in zero-sum games. Our next result shows that $\DRMplus$ with geometric discounting, that is, $\alpha^{(t)} = 1 - \gamma$ for some time-invariant $\gamma \in (0, 1)$, attains a rate of $T^{-1/2}$ to Nash equilibria in potential games; this is considerably faster than the $T^{-1/4}$ rate for alternating $\RMplus$ guaranteed by~\Cref{theorem:altern-RM+}. The basic reason is that $\DRMplus$---with geometric discounting---maintains the norm of the regret vector bounded by $\sqrt{ \nicefrac{m}{\gamma}}$ (\Cref{lemma:DRM} and~\Cref{cor:boundedreg}), while still enjoying the one-step improvement property of~\Cref{lemma:onestepRM+}.

\begin{restatable}{corollary}{DRMcor}
    \label{cor:discount-RM+}
    In any potential game, $\epsilon$-lazy alternating $\DRMplus$ with discount factor $1 - \gamma \in (0, 1)$ requires at most $1 + \frac{ m \Phimax }{\epsilon^2 \sqrt{\gamma} }$ rounds to converge to an $\epsilon$-Nash equilibrium.
\end{restatable}

To our knowledge, this is the first time that discounting yields a provable, worst-case improvement over the bound obtained for non-discounted $\RMplus$.

\paragraph{Regret matching} Before we switch gears to the more general constrained optimization setting, it is instructive to examine the behavior of $\RM$. It turns out that one can adjust~\Cref{lemma:onestepRM+}, but with a crucial caveat: the one-step improvement property is now only \emph{conditional}, as specified below.

\begin{restatable}{lemma}{onestepRM}
    \label{lemma:onestepRM}
    For any $\vR \in \R^\cA_{\geq 0}$ and $\vu \in \R^\cA$, we define $\vx \defeq \nicefrac{\vtheta}{\| \vtheta \|_1}$, where $\vtheta \defeq \max(\vR, \vec{0})$; if $\vtheta = \vec{0}$, $\vx \in \Delta(\cA)$ can be arbitrary. If $\vR' \defeq \vR + \vu - \langle \vx, \vu \rangle \vec{1}$ and $\vx' \defeq \nicefrac{\vtheta'}{\| \vtheta' \|_1}$, where $\vtheta' = \max(\vR', \vec{0}) \neq \vec{0}$, then
    \begin{equation}
        \label{eq:improve-RM}
        \langle \vx' - \vx, \vu \rangle \geq \frac{1}{\| \vtheta' \|_1} \|\vtheta' - \vtheta \|_2^2 \geq \frac{1}{\| \vtheta' \|_1} \left( \vu[a] - \langle \vx, \vu \rangle \right)^2 \mathbbm{1} \left\{ \vR[a] \geq 0 \right\},
    \end{equation}
    where $a \in \argmax_{a' \in \cA} \vu[a']$. If $\vtheta' = \vec{0}$, then $\langle \vx, \vu \rangle = \langle \vx', \vu \rangle \geq \vu[a]$.
\end{restatable}

We see that $\RM$'s one-step improvement is conditional on the regret accumulated thus far by a best-response action to be nonnegative. This is not an artifact of our analysis; it alludes to a fundamental discrepancy between $\RM$ and $\RMplus$ that will be formally established later on (\Cref{theorem:slowRM}). The main issue with $\RM$ can be seen as follows. If we consider a utility vector $\vu = (1, 0)$ and the initial regret vector is, say, $(-R, R)$, it will take $\RM$ many iterations---proportionally to the magnitude of $R > 0$---to finally change strategies, although this will eventually happen with a stationary utility.

\subsection{Constrained optimization and simultaneous updates}
\label{sec:nonlinear}

We now treat the more general setting where we are maximizing an $L$-smooth function $u$.

\paragraph{Single simplex} We begin with the special case of a single probability simplex, $\cX = \Delta(\cA)$. Our first goal is to adapt~\Cref{lemma:onestepRM+}. The key challenge is that $\RMplus$ does \emph{not} have a one-step improvement, unlike algorithms such as gradient descent (for a small enough learning rate), even if one initializes $\RMplus$ close to a KKT point. But we observe that if the norm of the regret vector is \emph{large enough}---having small regrets is an obstacle here, in contrast to~\Cref{sec:potential}---we are guaranteed a one-step improvement in terms of the value of the function (\Cref{lemma:R-L}).

To do so, we will use the basic quadratic bound, which yields $u(\vx') \geq u(\vx) + \langle \nabla u(\vx), \vx' - \vx \rangle - \frac{L}{2} \|\vx - \vx'\|_2^2$; we think of $\vx'$ as the updated strategy starting from $\vx$. Using a slight refinement of~\Cref{lemma:onestepRM+}, we first have the lower bound $\langle \vx' - \vx, \nabla u(\vx) \rangle \geq \frac{1}{\| \vR' \|_1} \|\vR - \vR' \|_2^2$ (\Cref{lemma:refinement}). 

Also, we observe that $\|\vx - \vx' \|_1 \leq \|\vR - \vR'\|_1 \left( \frac{1}{ \|\vR\|_1} + \frac{1}{\|\vR' \|_1} \right)$ (\Cref{lemma:closeness}). We are now ready to establish a \emph{conditional} one-step improvement when the regret vector has a \emph{sufficiently large} norm.

\begin{restatable}{lemma}{condiRMplus}
    \label{lemma:R-L}
    Let $u$ be an $L$-smooth function over $\Delta(\cA)$. For any $\vR \in \R^\cA_{\geq 0}$ with $\vr \neq \vec{0}$, we define $\vx \defeq \nicefrac{\vR}{\|\vR\|_1}$. Further, let $\vR' \defeq [ \vR + \nabla u(\vx) - \langle \vx, \nabla u(\vx) \rangle \vec{1} ]^+ \neq \vec{0}$ and $\vx' \defeq \nicefrac{\vR'}{ \|\vR'\|_1}$. If $\|\vR'\|_2 \geq \max\{ 2 m, 9 m L \}$, then
    \begin{equation*}
        u(\vx') - u(\vx) \geq \frac{1}{2 \|\vR'\|_1} \left( \max_{\vxstar \in \Delta(\cA)} \langle \vxstar - \vx, \nabla u(\vx)  \rangle  \right)^2.
    \end{equation*}
\end{restatable}
\Cref{lemma:R-L} only shows a one-step improvement so long as the norm of the regret vector is large enough. But how can we guarantee that? It would seem possible that $\RMplus$ ends up cycling in perpetuity under a regret vector with small norm. The following lemma shows that cannot happen; it turns out that maintaining this monotonicity property for the norm of the regret vector is crucial for designing an optimal variant of $\RMplus$ in zero-sum games~\citep{Zhang25:Scale}.

\begin{restatable}{lemma}{nondecreasingregret}
\label{prop:increasing-regret}
    For any $t$, $\RMplus$ guarantees $\| \vr^{(t)} \|_2^2 \ge \| \vr^{(t-1)} \|_2^2 + \| [\vg^{(t)}]^+ \|_2^2$, where $\vg^{(t)} \defeq \nabla u(\vx^{(t)}) - \langle \nabla u(\vx^{(t)}), \vx^{(t)} \rangle \vec{1}$ is the instantaneous regret vector at round $t$.
\end{restatable}
In particular,
\begin{equation*}
\| \vr^{(t)} \|_2^2 \ge \| \vr^{(t-1)} \|_2^2 + \| [\vg^{(t)}]_+ \|_2^2 \geq \| \vr^{(t-1)} \|_2^2 +  \KKTgap(\vx^{(t)})^2
\end{equation*}
since $\| [\vg^{(t)}]^+ \|_2^2 \geq \KKTgap(\vx^{(t)})^2$. That is, not only is the $\ell_2$ norm of the regret vector nondecreasing, but the increase is at least $\KKTgap(\vx^{(t)})^2$ at each round $t \in [T]$. Combining with~\Cref{lemma:R-L} yields the following.

\begin{restatable}[Single simplex]{theorem}{singlesimplex}
    \label{theorem:singlesimplex}
    Let $u $ be an $L$-smooth function in $\Delta(\cA) \subset \R^m$ with range $\urange$ and $R \defeq \max \{2 m, 9 m L \}$. $\RMplus$ requires at most $1 + \frac{ ( m ( 2 \urange + R^2 ) )^2 }{\epsilon^4} $ rounds to reach an $\epsilon$-KKT point.
\end{restatable}

\paragraph{Simultaneous updates in symmetric potential games}

We now use~\Cref{theorem:singlesimplex} to prove convergence of \emph{simultaneous} $\RMplus$ in \emph{symmetric} potential games; our earlier result in~\Cref{theorem:altern-RM+} shows convergence for arbitrary potential games but for the alternating version. The symmetry assumption here means that $\cA_1 = \cA_1 = \dots = \cA_n = \cA$ and $\vu_1(\vx_{-1}) = \vu_2(\vx_{-2}) = \dots = \vu_n( \vx_{-n} )$ when $\vx_1 = \vx_2 = \dots = \vx_n$. It is further assumed that all players initialize from the same strategy, so that the previous property implies that, inductively, it will be the case that $\vx_1^{(t)} = \vx_2^{(t)} = \dots = \vx_n^{(t)}$ for all $t$ under simultaneous updates because players observe exactly the same utility vector. A simple example of this is a two-player game with a common, symmetric payoff matrix $\mat{A} = \mat{A}^\top$. Then $\vu_1(\vx_2) = \mat{A} \vx_2$ and $\vu_2(\vx_1) = \mat{A} \vx_1$, so the previous assumption is satisfied.

\begin{restatable}{corollary}{sympot}
    \label{cor:sympot}
    In any symmetric potential game, simultaneous $\RMplus$ converges to an $\epsilon$-Nash equilibrium after $O_\epsilon(1/\epsilon^4)$ rounds. In particular, if convergence to the set of CCE happens at a rate of $T^{-(1 - \alpha)}$, for some $\alpha \in [0, \nicefrac{1}{2}]$, the rate of convergence to Nash equilibria is no worse than $T^{ - \frac{1-\alpha}{2}}$.
\end{restatable}

Indeed, running simultaneous $\RMplus$ in a symmetric game is equivalent to running $\RMplus$ on the function $\Delta(\cA) \ni \vx \mapsto \Phi(\vx, \dots, \vx)$, as we formalize in~\Cref{appendix:proofs} (\emph{cf.}~\citealp{Tewolde25:Computing}); \Cref{cor:sympot} thereby follows from~\Cref{theorem:singlesimplex} pertaining to~$\RMplus$ executed on a single simplex.

\paragraph{Multiple simplices} We now have the necessary tools to analyze the general case where we maximize $u$ over a product of simplices. Similarly to~\Cref{theorem:altern-RM+}, we run alternating $\RMplus$, thinking of every individual simplex as being controlled by a single player; this is akin to coordinate descent.

\begin{restatable}{corollary}{easycor}
    \label{cor:easycor}
    If $u$ is an $L$-smooth function in $\Delta(\cA_1) \times \dots \times \Delta(\cA_n)$ with range $\urange$, $\epsilon$-lazy alternating $\RMplus$ initialized at $\vr^{(0)}_i = \max\{ 2\sqrt{m_i}, 9 \sqrt{m_i} L \} \vec{1} $ for each player $i \in [n]$ requires at most $1 + \frac{ 4 n^4 m^2 \urange^2 }{\epsilon^4}$ rounds to reach an $\epsilon$-KKT point of $u$.
\end{restatable}

The proof follows directly from~\Cref{lemma:R-L} together with a telescopic summation. The non-lazy version of $\RMplus$ admits a qualitatively similar bound, following~\Cref{theorem:nonlazy}. Furthermore, a similar bound holds even under simultaneous $\RMplus$ (\Cref{cor:easycor-sim}), which follows by extending~\Cref{lemma:R-L} to multiple simplices (\Cref{lemma:R-L-multi}).

One caveat of those results is that the regret vector of each player needs to be initialized at a specific threshold. Our next result addresses this limitation by analyzing the usual parameter-free and scale-invariant version of $\RMplus$, at the cost of introducing a worse dependence on $1/\epsilon$.

\begin{restatable}{theorem}{mainth}
    \label{theorem:badrate}
    If $u$ is an $L$-smooth function in $\Delta(\cA_1) \times \dots \times \Delta(\cA_n)$ with range $\urange$, $\epsilon$-lazy alternating (or simultaneous) $\RMplus$ requires at most $O_\epsilon(1/\epsilon^8)$ rounds to reach an $\epsilon$-KKT point of $u$.
\end{restatable}
The proof again leverages~\Cref{prop:increasing-regret} to show that the number of rounds in which active players have a regret vector with norm smaller than what prescribed by~\Cref{lemma:R-L} is bounded by $O_\epsilon(1/\epsilon^2)$. It then appropriately bounds the number of rounds it takes in between such rounds to arrive at the bound of~\Cref{theorem:badrate}.

\section{Exponential lower bounds for regret matching}
\label{sec:RM}

In stark contrast, we show that $\RM$, with or without alternation, can take exponentially many rounds to reach an approximate Nash equilibrium even in two-player identical-interest games. The underlying class of games is based on the one considered by~\citet{Panageas23:Exponential}, who treated fictitious play. Specifically, for $m = 4, 6, \dots$ and $k \in \N$ we define the matrix $\mat{A}_{m, k}$ per the recursion
\begin{equation*}
    \R^{m \times m} \ni \mat{A}_{m, k} \defeq 
    \begin{bmatrix}
        k+1 & 0 & \cdots & 0 & 0\\
        0 & & & & k+4 \\
        \vdots & & \mat{A}_{m-2, k+4} & & \vdots \\
        0 & & & & 0\\
        k+2 & 0 & \cdots & 0 & k+3
    \end{bmatrix},
    \text{ where }
    \mat{A}_{2, k} \defeq 
    \begin{bmatrix}
        k+1 & 0 \\
        k+2 & k+3
    \end{bmatrix}.    
\end{equation*}
(An illustrative example appears in~\Cref{appendix:lowerbound}.) For any even dimension $m$, we define $\mat{A} \defeq \mat{A}_{m, 0}$, with maximum entry $2m - 1$. Further, we define, for $1 \leq a_1 \leq m+1$ and $1 \leq a_2 \leq m+1$,
\begin{equation}
    \label{eq:gamelower}
    \R^{(m+1) \times (m+1)} \ni \mat{B}[a_1, a_2] \defeq 
    \begin{cases}
        \mat{A}[a_1, a_2] & \text{if } a_1 \leq m \text{ and } a_2 \leq m;\\
        1/2 & \text{if } (a_1 = m+1 \text{ and } a_2 = 1) \text{ or } (a_1 = 1 \text{ and } a_2 = m+1); \\
        0 & \text{otherwise}.
    \end{cases}
\end{equation}
The action sets of the two players are $\cA_1 = [m+1] = \cA_2$. We assume that $\RM$ is initialized at the pure strategy $(m+1, m+1)$; \Cref{appendix:lowerbound} shows how to adapt the lower bound when $\RM$ is initialized at the uniform random strategy (\Cref{cor:uniform-init}), which is more common. We recall that one round includes one update from each player, which for now is assumed to be made in a simultaneous fashion. For a payoff $k \in \N$, we denote by $a_1(k), a_2(k) \in [m]$ the row and column index, respectively, corresponding to $k$ in the matrix $\mat{A}$.

We begin by stating a basic invariance concerning the behavior of $\RM$ when executed on the game~\eqref{eq:gamelower}.

\begin{restatable}{property}{invariance}
    \label{lemma:resting-transition}
    After the first round both players play the first action. Thereupon, either the players play with probability $1$ $(a_1(k), a_2(k))$, or, when $k$ is odd, only Player 1 (respectively, Player 2 when $k$ is even) mixes between $a_1(k)$ and $a_1(k+1)$ (respectively, $a_2(k)$ and $a_2(k+1)$). If a row or a column stops being played, it will never be played henceforth. An action profile $(a_1(k+1), a_2(k+1))$ is played with positive probability only if $(a_1(k), a_2(k))$ was played at some previous round.
\end{restatable}

We prove this property inductively in~\Cref{appendix:lowerbound}. We take it for granted in what follows.

In accordance with~\Cref{lemma:resting-transition}, for $k \geq 2$, we define $\tlow{k}$ to be the first round in which the action profile corresponding to payoff $k$ is played with positive probability and $\thigh{k}$ the last round before the action profile corresponding to payoff $k+1$ is played with positive probability. We then define $T_k \defeq \thigh{k} - \tlow{k} + 1$ to be the number of rounds corresponding to the period $[\tlow{k}, \thigh{k}]$.

We also define $\cA_1(k) \defeq \{ a_1(k') : 2m - 1 \geq k' \geq k \}$ and $\cA_2(k) \defeq \{ a_2 (k') : 2m-1 \geq k' \geq k \}$. These are the rows and columns, respectively, that will be played after the action profile corresponding to $k$ starts being played. The next crucial lemma shows that before an action becomes desirable, it will have accumulated very negative regret in the previous rounds.

\begin{lemma}
    \label{lemma:negreg}
    For any even $k \geq 4$, let $\vr_1^{(\thigh{k-2})}[a_1]$ be the regret of Player 1 with respect to any action $a_1 \in \cA_1(k)$. Then $\vr_{1}^{(\thigh{k-2} )}[a_1] \leq - \sum_{l=2}^{k-2} (l-1) T_l$. Similarly, for any odd $k \geq 5$, if $\vr_2^{(\thigh{k-2})}[a_2]$ is the regret of Player 2 with respect to any action $a_2 \in \cA_2(k)$, $\vr_{2}^{(\thigh{k-2} )}[a_2] \leq - \sum_{l=2}^{k-2} (l-1) T_l$.
\end{lemma}

At the same time, when an action has very negative regret, it will take a long time before that action gets played with positive probability, as formalized below.

\begin{lemma}
    \label{lemma:overcomingnegregret}
    For any even $k \geq 4$, $T_k \geq - \frac{1}{2} \vr_{2}^{(\thigh{k-1})}[a_2(k+1)]$. Similarly, for every odd $k \geq 5$, $T_k \geq - \frac{1}{2} \vr_{1}^{(\thigh{k-1})}[a_1(k+1)]$.
\end{lemma}

By~\Cref{lemma:negreg,lemma:overcomingnegregret}, it follows that $T_k \geq \sum_{l=2}^{k-1} \frac{l - 1}{2} T_l$ for any $k \geq 4$. By the inductive basis, we know that $T_3 \geq 1$. As a result, $T_k \geq \frac{k-2}{2} T_{k-1} \geq \frac{k-2}{2} \frac{k-3}{2} \dots \frac{2}{2} T_3 \geq \frac{(k-2)!}{2^{k-3}}$ for all $k \geq 4$.

Moreover, it takes as least $T_{2m - 2}$ rounds to converge to an NE with approximation gap at most $\nicefrac{1}{2m + 2}$ (\Cref{lemma:mixedNE}). We thus arrive at the following exponential lower bound.

\begin{theorem}
    \label{theorem:slowRM}
    Simultaneous $\RM$ requires $m^{\Omega(m)}$ rounds to converge to a $\frac{1}{2m}$-Nash equilibrium in two-player $m \times m$ identical-interest games.
\end{theorem}

The same reasoning directly applies to alternating $\RM$.

\begin{corollary}
    Alternating $\RM$ requires $m^{\Omega(m)}$ rounds to converge to a $\frac{1}{2m}$-Nash equilibrium in two-player $m \times m$ identical-interest games.
\end{corollary}

\section{Future research}

Our paper sheds new light on the convergence properties of regret matching($^+$) in constrained optimization problems in general, and potential games in particular. We showed that $\RMplus$ is a sound and fast first-order optimizer; on the flip side, $\RM$ can be exponentially slow even in two-player identical-interest games.

Several interesting questions remain open. It would be interesting to understand whether $\RMplus$, with or without alternation, can experience $\Omega(\sqrt{T})$ regret in potential games; this is known to be the case in zero-sum games~\citep{Farina23:Regret}, but remains unclear for the class of potential games. In light of our results, any improvement over the $\sqrt{T}$ barrier would automatically translate into a faster convergence rate to Nash equilibria. Moreover, we have worked exclusively in the full feedback setting; extending our results under stochastic or bandit feedback would be a natural next step. Finally, does $\RM$ asymptotically converge even under alternating updates?

\section*{Acknowledgments}

Emanuel Tewolde and Vincent Conitzer thank the Cooperative Al Foundation, Macroscopic Ventures and Jaan Tallinn's donor-advised fund at Founders Pledge for financial support.
Emanuel Tewolde is also supported in part by the Cooperative AI PhD Fellowship. Ioannis Panageas is supported by NSF grant CCF-2454115. Tuomas Sandholm is supported by the Vannevar Bush Faculty Fellowship ONR N00014-23-1-2876, National Science Foundation grants RI-2312342 and RI-1901403, ARO award W911NF2210266, and NIH award A240108S001.

\bibliography{main}

@String{AAAI = "Conference on Artificial Intelligence (AAAI)"}

@String{ICLR = "International Conference on Learning Representations (ICLR)"}

@String{AAMAS =     "Autonomous Agents and Multi-Agent Systems (AAMAS)"}

@String{AI =     "Artificial Intelligence"}

@String{C =      "IEEE Transactions on Computers"}

@String{GEB =     "Games and Economic Behavior"}

@String{ICML = "International Conference on Machine Learning (ICML)"}

@String{IJCAI = "Proceedings of the International Joint
                 Conference on Artificial Intelligence (IJCAI)"}

@String{JACM =   "Journal of the ACM"}

@String{NIPS = "Neural Information
                  Processing Systems (NIPS)"}

@String{NeurIPS = "Neural Information
                  Processing Systems (NeurIPS)"}

@String{PODC = "Proceedings of the ACM Symposium on Principles of Distributed Computing"}

@String{SARA = "Symposium on Abstraction, Reformulation and Approximation (SARA)"}

@String{SIAM =    "SIAM Journal on Computing"}

@String{STOC = "Proceedings of the Annual Symposium on Theory of
Computing (STOC)"}

@String{UAI = "Proceedings of the Conference on Uncertainty in Artificial
                Intelligence (UAI)"}

@string{jan = "January"}

@string{may = "May"}

@string{oct = "October"}

@inproceedings{Tewolde25:Decision,
    author = {Emanuel Tewolde and
              Brian Hu Zhang and
              Ioannis Anagnostides and
              Tuomas Sandholm and
              Vince Conitzer},
    title = {Decision Making under Imperfect Recall: Algorithms and Benchmarks},
    booktitle = {Uncertainty in Artificial Intelligence (UAI)},
    year = 2025
}

@inproceedings{Panageas23:Exponential,
  author       = {Ioannis Panageas and
                  Nikolas Patris and
                  Stratis Skoulakis and
                  Volkan Cevher},
  title        = {Exponential Lower Bounds for Fictitious Play in Potential Games},
  booktitle    = NeurIPS,
  year         = {2023}
}

@inproceedings{Babichenko21:Settling,
  author       = {Yakov Babichenko and
                  Aviad Rubinstein},
  title        = {Settling the complexity of {N}ash equilibrium in congestion games},
  booktitle    = STOC,
  year         = {2021}
}

@inproceedings{Wibisono22:Alternating,
  author       = {Andre Wibisono and
                  Molei Tao and
                  Georgios Piliouras},
  title        = {Alternating Mirror Descent for Constrained Min-Max Games},
  booktitle    = {Neural Information Processing Systems},
  year         = {2022}
}

@inproceedings{Xu24:Dynamic,
  author       = {Hang Xu and
                  Kai Li and
                  Haobo Fu and
                  Qiang Fu and
                  Junliang Xing and
                  Jian Cheng},
  title        = {Dynamic Discounted Counterfactual Regret Minimization},
  booktitle    = {International Conference on Learning Representations (ICLR)},
  year         = {2024}
}

@article{Zhang24:Faster,
      title={Faster Game Solving via Hyperparameter Schedules}, 
      author={Naifeng Zhang and Stephen McAleer and Tuomas Sandholm},
      year={2024},
      journal={arXiv:2404.09097}
}

@inproceedings{Cevher23:Alternation,
  author       = {Volkan Cevher and
                  Ashok Cutkosky and
                  Ali Kavis and
                  Georgios Piliouras and
                  Stratis Skoulakis and
                  Luca Viano},
  title        = {Alternation makes the adversary weaker in two-player games},
  booktitle    = {Neural Information Processing Systems (NeurIPS)},
  year         = {2023}
}

@article{Hoeffding58:Distinguishability,
 author = {Wassily Hoeffding and J. Wolfowitz},
 journal = {The Annals of Mathematical Statistics},
 number = {3},
 pages = {700--718},
 publisher = {Institute of Mathematical Statistics},
 title = {Distinguishability of Sets of Distributions},
 volume = {29},
 year = {1958}
}

@inproceedings{Tewolde25:Computing,
  author       = {Emanuel Tewolde and
                  Brian Hu Zhang and
                  Caspar Oesterheld and
                  Tuomas Sandholm and
                  Vincent Conitzer},
  title        = {Computing Game Symmetries and Equilibria That Respect Them},
  booktitle    = AAAI,
  year         = {2025}
}

@inproceedings{Gimbert20:Bridge,
  author       = {Hugo Gimbert and
                  Soumyajit Paul and
                  B. Srivathsan},
  title        = {A Bridge between Polynomial Optimization and Games with Imperfect
                  Recall},
  booktitle    = AAMAS,
  year         = {2020}
}

@inproceedings{Berker25:Value,
  author       = {Ratip Emin Berker and
                  Emanuel Tewolde and
                  Ioannis Anagnostides and
                  Tuomas Sandholm and
                  Vincent Conitzer},
  title        = {The Value of Recall in Extensive-Form Games},
  booktitle    = AAAI,
  year         = {2025}
}

@inproceedings{Kleinberg09:Multiplicative,
  author       = {Robert Kleinberg and
                  Georgios Piliouras and
                  {\'{E}}va Tardos},
  title        = {Multiplicative updates outperform generic no-regret learning in congestion
                  games: extended abstract},
  booktitle    = STOC,
  year         = {2009}
}

@inproceedings{Marden07:Regret,
  author       = {Jason R. Marden and
                  G{\"{u}}rdal Arslan and
                  Jeff S. Shamma},
    title        = {Regret based dynamics: convergence in weakly acyclic games},
  booktitle    = AAMAS,
  year         = {2007}
}

@article{Ma14:Distributed,
  title={Distributed regret matching algorithm for dynamic congestion games with information provision},
  author={Ma, Tai-Yu and Gerber, Philippe},
  journal={Transportation Research Procedia},
  volume={3},
  pages={3--12},
  year={2014},
  publisher={Elsevier}
}

@article{Aumann74:Subjectivity,
  author  = {Robert Aumann},
  title   = {Subjectivity and Correlation in Randomized Strategies},
  journal = {Journal of Mathematical Economics},
  year    = {1974},
  volume  = {1},
  pages   = {67--96}
}

@article{Nash50:Equilibrium,
  author  = {John Nash},
  title   = {Equilibrium points in N-person games},
  journal = {Proceedings of the National Academy of Sciences},
  year    = {1950},
  volume  = {36},
  pages   = {48--49}
}

@inproceedings{Ivgi23:Parameter,
  author       = {Maor Ivgi and
                  Oliver Hinder and
                  Yair Carmon},
  title        = {DoG is {SGD}'s Best Friend: {A} Parameter-Free Dynamic Step Size Schedule},
  booktitle    = {International Conference on Machine Learning (ICML)},
  year         = {2023}
}

@inproceedings{Orabona16:Coin,
  author       = {Francesco Orabona and
                  D{\'{a}}vid P{\'{a}}l},
  title        = {Coin Betting and Parameter-Free Online Learning},
  booktitle    = NIPS,
  year         = {2016}
}

@inproceedings{Defazio23:Learning,
  author       = {Aaron Defazio and
                  Konstantin Mishchenko},
  title        = {Learning-Rate-Free Learning by D-Adaptation},
  booktitle    = {International Conference on Machine Learning (ICML)},
  year         = {2023}
}

@inproceedings{Heliou17:Learning,
  author       = {Am{\'{e}}lie H{\'{e}}liou and
                  Johanne Cohen and
                  Panayotis Mertikopoulos},
  title        = {Learning with Bandit Feedback in Potential Games},
  booktitle    = NeurIPS,
  year         = {2017}
}

@inproceedings{Panageas23:Semi,
  title={Semi Bandit Dynamics in Congestion Games: Convergence to {N}ash Equilibrium and No-Regret Guarantees},
  author={Panageas, Ioannis and Skoulakis, Stratis and Viano, Luca and Wang, Xiao and Cevher, Volkan},
  booktitle={International Conference on Machine Learning (ICML)},
  year={2023}
}

@inproceedings{Cui22:Learning,
  title={Learning in congestion games with bandit feedback},
  author={Cui, Qiwen and Xiong, Zhihan and Fazel, Maryam and Du, Simon S},
  booktitle=NeurIPS,
  year={2022}
}

@article{Fearnley22:Complexity,
  title={The complexity of gradient descent: {CLS}= {PPAD} $\cap$ {PLS}},
  author={Fearnley, John and Goldberg, Paul and Hollender, Alexandros and Savani, Rahul},
  journal={Journal of the ACM},
  volume={70},
  number={1},
  pages={1--74},
  year={2022}
}

@article{Nan25:Convergence,
      title={On the ${O}(1/{T})$ Convergence of Alternating Gradient Descent-Ascent in Bilinear Games}, 
      author={Tianlong Nan and Shuvomoy Das Gupta and Garud Iyengar and Christian Kroer},
      year={2025},
      journal={arXiv:2510.03855}
}

@article{Zhang25:Scale,
      title={Scale-Invariant Regret Matching and Online Learning with Optimal Convergence: Bridging Theory and Practice in Zero-Sum Games}, 
      author={Brian Hu Zhang and Ioannis Anagnostides and Tuomas Sandholm},
      year={2025},
      journal = {arXiv:2510.04407}
}

@article{Angelopoulos25:Gradient,
  author       = {Anastasios N. Angelopoulos and
                  Michael I. Jordan and
                  Ryan J. Tibshirani},
  title        = {Gradient Equilibrium in Online Learning: Theory and Applications},
  journal      = {arXiv:2501.08330},
  year         = {2025}
}

@inproceedings{Hazan17:Efficient,
  author       = {Elad Hazan and
                  Karan Singh and
                  Cyril Zhang},
  title        = {Efficient Regret Minimization in Non-Convex Games},
  booktitle    = {International Conference on Machine Learning (ICML)},
  year         = {2017}
}

@inproceedings{Palaiopanos17:Multiplicative,
  author       = {Gerasimos Palaiopanos and
                  Ioannis Panageas and
                  Georgios Piliouras},
  title        = {Multiplicative Weights Update with Constant Step-Size in Congestion
                  Games: Convergence, Limit Cycles and Chaos},
  booktitle    = NeurIPS,
  year         = {2017}
}

@inproceedings{Blum06:Routing,
  author    = {Avrim Blum and Eyal Even-Dar and Katrina Ligett},
  title     = {Routing without Regret: On Convergence to {N}ash Equilibria of Regret-Minimizing Algorithms in Routing Games},
  booktitle = PODC,
  year      = {2006}
}

@article{Moravvcik17:DeepStack,
  author  = {Morav{\v c}{\'i}k, Matej and Schmid, Martin and Burch, Neil and Lis{\'y}, Viliam and Morrill, Dustin and Bard, Nolan and Davis, Trevor and Waugh, Kevin and Johanson, Michael and Bowling, Michael},
  title   = {DeepStack: Expert-level artificial intelligence in heads-up no-limit poker},
  journal = {Science},
  year    = {2017},
  month   = {May}
}

@inproceedings{Tewolde23:Computational,
  author       = {Emanuel Tewolde and
                  Caspar Oesterheld and
                  Vincent Conitzer and
                  Paul W. Goldberg},
  title        = {The Computational Complexity of Single-Player Imperfect-Recall Games},
  booktitle    = IJCAI,
  year         = {2023}
}

@article{Monderer96:Potential,
  title={Potential games},
  author={Monderer, Dov and Shapley, Lloyd S},
  journal=GEB,
  volume={14},
  number={1},
  pages={124--143},
  year={1996}
}

@article{Freund99:Adaptive,
  author  = {Yoav Freund and Robert Schapire},
  title   = {Adaptive game playing using multiplicative weights},
  journal = GEB,
  year    = {1999},
  volume  = {29},
  pages   = {79--103}
}

@article{Hart00:Simple,
  author  = {Sergiu Hart and Andreu Mas-Colell},
  title   = {A Simple Adaptive Procedure Leading to Correlated Equilibrium},
  journal = {Econometrica},
  year    = {2000},
  volume  = {68},
  pages   = {1127--1150}
}

@article{Chen09:Settling,
  author  = {Xi Chen and Xiaotie Deng and Shang-Hua Teng},
  title   = {Settling the Complexity of Computing Two-Player {N}ash Equilibria},
  journal = JACM,
  year    = {2009}
}

@article{Rosenthal73:Class,
  title={A class of games possessing pure-strategy Nash equilibria},
  author={Rosenthal, Robert W},
  journal={International Journal of Game Theory},
  volume={2},
  number={1},
  pages={65--67},
  year={1973}
}

@article{Daskalakis08:Complexity,
  author  = {Constantinos Daskalakis and Paul Goldberg and Christos Papadimitriou},
  title   = {The Complexity of Computing a {N}ash Equilibrium},
  journal = SIAM,
  year    = {2008}
}

@inproceedings{Waugh09:Practical,
  author    = {Kevin Waugh and Martin Zinkevich and Michael Johanson and Morgan Kan and David Schnizlein and Michael Bowling},
  title     = {A Practical Use of Imperfect Recall},
  booktitle = {Symposium on Abstraction, Reformulation and Approximation (SARA)},
  year      = {2009}
}

@article{Koller92:Complexity,
  author  = {Daphne Koller and Nimrod Megiddo},
  title   = {The Complexity of Two-Person Zero-Sum Games in Extensive Form},
  journal = {Games and Economic Behavior},
  year    = {1992},
  volume  = {4},
  number  = {4},
  pages   = {528--552},
  month   = oct
}

@article{Xu24:Minimizing,
      title={Minimizing Weighted Counterfactual Regret with Optimistic Online Mirror Descent}, 
      author={Hang Xu and Kai Li and Bingyun Liu and Haobo Fu and Qiang Fu and Junliang Xing and Jian Cheng},
      year={2024},
      journal={arXiv:2404.13891}
}

@inproceedings{Chakrabarti24:Extensive,
 author = {Chakrabarti, Darshan and Grand-Cl\'{e}ment, Julien and Kroer, Christian},
 booktitle = NeurIPS,
 title = {Extensive-Form Game Solving via Blackwell Approachability on Treeplexes},
 year = {2024}
}

@article{Meng25:Asynchronous,
      title={Asynchronous Predictive Counterfactual Regret Minimization$^+$ Algorithm in Solving Extensive-Form Games}, 
      author={Linjian Meng and Youzhi Zhang and Zhenxing Ge and Tianpei Yang and Yang Gao},
      year={2025},
      journal={arXiv:2503.12770}
}

@inproceedings{Cai24:Fast,
  author       = {Yang Cai and
                  Gabriele Farina and
                  Julien Grand{-}Cl{\'{e}}ment and
                  Christian Kroer and
                  Chung{-}Wei Lee and
                  Haipeng Luo and
                  Weiqiang Zheng},
  title        = {Fast Last-Iterate Convergence of Learning in Games Requires Forgetful
                  Algorithms},
  booktitle    = NeurIPS,
  year         = {2024}
}

@inproceedings{Cai25:Last,
  author       = {Yang Cai and
                  Gabriele Farina and
                  Julien Grand{-}Cl{\'{e}}ment and
                  Christian Kroer and
                  Chung{-}Wei Lee and
                  Haipeng Luo and
                  Weiqiang Zheng},
  title        = {Last-Iterate Convergence Properties of Regret-Matching Algorithms
                  in Games},
  booktitle    = {International Conference on Learning Representations (ICLR)},
  year         = {2025}
}

@article{Moulin78:Strategically,
  author  = {H. Moulin and J.-P. Vial},
  title   = {Strategically zero-sum games: The class of games whose completely mixed equilibria cannot be improved upon},
  journal = {International Journal of Game Theory},
  year    = {1978},
  volume  = {7},
  number  = {3-4},
  pages   = {201--221}
}

@article{Piccione97:Interpretation,
  title={On the interpretation of decision problems with imperfect recall},
  author={Piccione, Michele and Rubinstein, Ariel},
  journal=GEB,
  pages={3--24},
  year={1997},
  publisher={Elsevier}
}

@article{Hart03:Regret,
  title={Regret-based continuous-time dynamics},
  author={Hart, Sergiu and Mas-Colell, Andreu},
  journal={Games and Economic Behavior},
  volume={45},
  number={2},
  pages={375--394},
  year={2003},
  publisher={Elsevier}
}

@inproceedings{Zinkevich07:Regret,
  author    = {Martin Zinkevich and Michael Bowling and Michael Johanson and Carmelo Piccione},
  title     = {Regret Minimization in Games with Incomplete Information},
  booktitle = NIPS,
  year      = {2007}
}

@article{Tammelin14:Solving,
  title         = {Solving Large Imperfect Information Games Using {CFR+}},
  author        = {Oskari Tammelin},
  year          = {2014},
  journal        = {arXiv:1407.5042}
}

@article{Bowling15:Heads,
  author  = {Michael Bowling and Neil Burch and Michael Johanson and Oskari Tammelin},
  title   = {Heads-up Limit Hold'em Poker is Solved},
  journal = {Science},
  year    = {2015},
  volume  = {347},
  number  = {6218},
  month   = jan
}

@article{Blackwell56:analog,
  author  = {David Blackwell},
  title   = {An analog of the minmax theorem for vector payoffs},
  journal = {Pacific Journal of Mathematics},
  year    = {1956},
  volume  = {6},
  pages   = {1--8}
}

@article{Brown17:Superhuman,
  title     = {Superhuman {AI} for heads-up no-limit poker: {Libratus} beats top professionals},
  author    = {Brown, Noam and Sandholm, Tuomas},
  journal   = {Science},
  pages     = {eaao1733},
  year      = {2017},
  month     = {Dec.},
  publisher = {American Association for the Advancement of Science}
}

@inproceedings{Brown19:Solving,
  title     = {Solving imperfect-information games via discounted regret minimization},
  author    = {Noam Brown and Tuomas Sandholm},
  booktitle = AAAI,
  year      = {2019}
}

@article{Brown19:Superhuman,
  author  = {Brown, Noam and Sandholm, Tuomas},
  title   = {Superhuman {AI} for multiplayer poker},
  year    = {2019},
  journal = {Science},
  volume  = {365},
  number  = {6456},
  pages   = {885--890}
}

@inproceedings{Farina21:Faster,
  title     = {Faster Game Solving via Predictive {B}lackwell Approachability: Connecting Regret Matching and Mirror Descent},
  author    = {Farina, Gabriele and Kroer, Christian and Sandholm, Tuomas},
  booktitle = AAAI,
  year      = {2021}
}

@inproceedings{Farina23:Regret,
  title={Regret matching+:(in) stability and fast convergence in games},
  author={Farina, Gabriele and Grand-Cl{\'e}ment, Julien and Kroer, Christian and Lee, Chung-Wei and Luo, Haipeng},
  booktitle=NeurIPS,
  year={2023}
}

@article{Zhang25:General,
      title={General search techniques without common knowledge for imperfect-information games, and application to superhuman Fog of War chess}, 
      author={Brian Hu Zhang and Tuomas Sandholm},
      year={2025},
      journal={arXiv:2506.01242}
}

\clearpage

\appendix

\section{Further related work}
\label{sec:related}

Much of the existing research on regret matching revolves around zero-sum games. Many variants have been proposed over the years to speed up its convergence~\citep{Xu24:Minimizing,Cai25:Last,Chakrabarti24:Extensive,Meng25:Asynchronous,Farina21:Faster,Tammelin14:Solving,Brown19:Solving}. Some notable variations that have considerably improved performance are \emph{predictive} $\RM$ and $\RMplus$~\citep{Farina21:Faster}, which rely on predicting the next utility, and \emph{discounted} $\RM$ and $\RMplus$~\citep{Brown19:Solving,Zhang24:Faster,Xu24:Dynamic}, where one dynamically discounts the accumulated regret; in a similar vein, our work shows that a discounted variant of $\RMplus$ achieves a better convergence upper bound than $\RMplus$ in our setting (\Cref{cor:discount-RM+}). It must be stressed that the focus of all that prior work was on zero-sum games. Constrained optimization is a fundamentally different problem. For one, in zero-sum games, it is only the average strategy of $\RM$ and $\RMplus$ that converges, not the last iterate~\citep{Farina23:Regret}.

The recent paper of~\citet{Tewolde25:Decision} demonstrated that the regret matching family is a formidable first-order optimizer in constrained optimization problems. In particular, their focus was on (single-player) imperfect-recall problems, which are tantamount to general polynomial optimization problems over a product of simplices. Interestingly, many of the trends observed in zero-sum games are actually reversed in constrained optimization. For example, the predictive versions of $\RM$ and $\RMplus$ generally performed worse than their non-predictive counterparts. One trend that did persist was the superiority of $\RMplus$ over $\RM$. It is also worth mentioning an earlier work by~\citet{Ma14:Distributed} that also reported fast empirical convergence in a certain class of congestion games. Yet, there was hitherto no theoretical understanding of those algorithms in this setting. The main precursors of our work are the paper of \citet{Hart03:Regret}, which established asymptotic convergence in discrete time but for a somewhat artificial variant of regret matching, and the paper of~\citet{Marden07:Regret}, which analyzed asymptotically a certain variant of regret matching that aggressively discounts the regrets (\emph{cf.}~\Cref{cor:discount-RM+}).

An interesting result that sheds light on $\RM$ and $\RMplus$ is by~\citet{Farina21:Faster}, who showed that $\RM$ can be obtained by running \emph{follow the regularized leader (\ftrl)} in a certain lifted space, whereas $\RMplus$ can be obtained through \emph{mirror descent (\md)} in the same space; this is despite the fact that, unlike $\ftrl$ and $\md$, $\RM$ and $\RMplus$ are both parameter free. On a related note, \citet{Cai24:Fast} showed that only forgetful algorithms---closer to $\md$ than to $\ftrl$---can attain fast last-iterate convergence. Our exponential separation of $\RM$ and $\RMplus$ echoes their finding, although in a different setting and class of algorithms.

Zooming out of the $\RM$ family, understanding the convergence of no-regret dynamics in potential games has been a popular research topic~\citep{Kleinberg09:Multiplicative,Heliou17:Learning,Palaiopanos17:Multiplicative,Panageas23:Semi,Cui22:Learning,Blum06:Routing}. Our research also relates to parameter-free optimization; for example, we refer to~\citet{Ivgi23:Parameter,Orabona16:Coin,Defazio23:Learning} and references therein.
\section{Further background}
\label{appendix:prels}

\paragraph{Coarse correlated equilibria} For completeness, we provide the definition of a coarse correlated equilibrium~\citep{Moulin78:Strategically}, which is a relaxation of correlated equilibria~\citep{Aumann74:Subjectivity}. The key connection that relates to our results is that if all players in a normal-form game have sublinear regret, the average correlated distribution of play converges to the set of coarse correlated equilibria. In particular, the rate of convergence is driven by the maximum of the players' regrets (\Cref{prop:CCE}).

\begin{definition}[Coarse correlated equilibrium]
    \label{def:CCE}
    Consider an $n$-player game in normal form. A correlated distribution $\mu \in \Delta(\cA_1 \times \dots \times \cA_n)$ is an \emph{$\epsilon$-coarse correlated equilibrium (CCE)} if for any player $i \in [n]$ and deviation $a_i' \in \cA_i$,
    \begin{equation*}
        \E_{(a_1, \dots, a_n) \sim \mu} u_i(a_1, \dots, a_n) \geq \E_{(a_1, \dots, a_n) \sim \mu } u_i(a_i', a_{-i}) - \epsilon.
    \end{equation*}
\end{definition}

\begin{proposition}
    \label{prop:CCE}
    If each player $i \in [n]$ observes the sequence of utilities $(\vu_i(\vx^{(t)}_{-i}))_{t=1}^T$, the average correlated distribution of play is an $\epsilon$-CCE with $\epsilon \leq \frac{1}{T} \max_{1 \leq i \leq n} \reg_i^{(T)}$, where $\reg_i^{(T)}$ is the regret of the $i$th player.
\end{proposition}
This connection holds for simultaneous updates. It is unclear if and how it can be extended under alternating updates. For the special case of potential games and $\RMplus$, which is our main focus here, we are indeed able to establish convergence to the set of CCEs even under alternating updates by bounding the path length of the players' strategies (\Cref{remark:alternating-CCE}).

\paragraph{Other notions of regret} \Cref{sec:onlinelearning} introduced the usual notion of regret used in online linear optimization. For a constrained optimization problem with respect to a differentiable function $u$, we have $\reg^{(T)} = \max_{\vx' \in \cX} \sum_{t=1}^T \langle \vx' - \vx^{(t)}, \nabla_\vx u(\vx^{(t)}) \rangle$. This is a linearized version of $\max_{\vx' \in \cX} \sum_{t=1}^T ( u(\vx') - u(\vx^{(t)}))$. Minimizing the latter notion is computationally intractable unless one places restrictive assumptions on $u$. \citet{Hazan17:Efficient} (\emph{cf.}~\citealp{Angelopoulos25:Gradient}) introduced the notion of ``local regret,'' and showed that having sublinear local regret implies that a randomly selected iterate will be an approximate stationary point; this is in stark contrast to regret as defined in~\Cref{sec:onlinelearning} (\Cref{prop:4cycle}). 

\paragraph{Discounting} Next, we spell out regret matching$^+$ with discounting ($\DRMplus$; \Cref{alg:regret_matching_plus-discounting}). The only difference from $\RMplus$ is that the regret vector is multiplied by a discounting coefficient $\alpha^{(t)} \in (0, 1]$ in every round $t \in [T]$ (\Cref{line:discount}); the special case where $\alpha^{(t)} = 1$ for all $t \in [T]$ is $\RMplus$.

\begin{algorithm}[H]
\caption{Regret matching$^+$ with discounting ($\DRMplus$)}
\label{alg:regret_matching_plus-discounting}
\textbf{Input}: discounting coefficients $(\alpha^{(1)}, \dots, \alpha^{(T)}) \in (0, 1]^T $\;
\SetKwInOut{Input}{Input}
\SetKwInOut{Output}{Output}
Initialize cumulative regrets $\vR^{(0)} \defeq \vec{0}$\;
Initialize strategy $\vx^{(1)} \in \Delta(\cA)$\;
\For{$t = 1, \dots, T$}{
    Set $\vtheta^{(t)} \leftarrow \vR^{(t-1)}$\;
    \If{$ \vtheta^{(t)} \neq \vec{0}$} {
    Compute $\vx^{(t)} \leftarrow \nicefrac{\vtheta^{(t)}}{\| \vtheta^{(t)} \|_1 } $\;
    }
    \Else{
        $\vx^{(t)} \leftarrow \vx^{(t-1)}$\;
    }
    Output strategy $\vx^{(t)} \in \Delta(\cA)$ \;
    Observe utility $\vu^{(t)} \in \R^{\cA}$\;
    $\vR^{(t)} \leftarrow \alpha^{(t)} [\vR^{(t-1)} + \vu^{(t)} - \langle \vx^{(t)}, \vu^{(t)} \rangle \vec{1}]^+$\;\label{line:discount}
}
\end{algorithm}
\section{Omitted proofs}
\label{appendix:proofs}

This section provides the proofs missing from the main body. We begin by stating a simple lemma that bounds the regret of $\RMplus$, implying~\Cref{prop:noregret-RM}; we will then adapt it to account for discounting per~\Cref{alg:regret_matching_plus-discounting}.

\begin{lemma}[Regret vector upper bound]\label{prop:regret bound}
    For any time $t \in [T]$, $\RMplus$ guarantees $\| \vr^{(t)} \|_2^2 \le  \| \vr^{(t-1)} \|_2^2 + \| \vg^{(t)} \|_2^2 $, where $\vg^{(t)} \defeq \vu^{(t)} - \langle \vx^{(t)}, \vu^{(t)} \rangle$ is the instantaneous regret at time $t$.
\end{lemma}

\begin{proof}
    By definition of $\RMplus$,  $\langle \vr^{(t-1)}, \vg^{(t)} \rangle = \langle \vx^{(t)}, \vg^{(t)} \rangle  = 0$ since $\vx^{(t)} \propto \vr^{(t-1)}$. Thus, $$\| \vr^{(t)} \|_2^2 = \| [\vr^{(t-1)} + \vg^{(t)}]^+ \|_2^2 \le \|\vr^{(t-1)} + \vg^{(t)} \|_2^2 = \| \vr^{(t-1)} \|_2^2 + \| \vg^{(t)} \|_2^2,$$
    by orthogonality. 
\end{proof}
As a result, the telescopic summation yields $\| \vr^{(T)} \|_2^2 \leq \sum_{t=1}^T \| \vg^{(t)} \|_2^2 \leq m T$ since $\| \vg^{(t)} \|_\infty \leq 1$ (by the assumption that the range of the utilities is bounded by $1$). It is worth noting that a similar proof works for $\RM$. We now adapt~\Cref{prop:regret bound} for $\DRMplus$.

\begin{lemma}
    \label{lemma:DRM}
    For any time $t \in [T]$, $\DRMplus$ guarantees $\| \vr^{(t)} \|_2^2 \le (\alpha^{(t)} )^2 ( \| \vr^{(t-1)} \|_2^2 + \| \vg^{(t)} \|_2^2 )$.
\end{lemma}

\begin{proof}
    As before, $\langle \vr^{(t-1)}, \vg^{(t)} \rangle = \langle \vx^{(t)}, \vg^{(t)} \rangle = 0$ since $\vx^{(t)} \propto \vr^{(t-1)}$. Thus,
    \begin{equation*}
        \| \vr^{(t)} \|_2^2 = ( \alpha^{(t)})^2 \| [\vr^{(t-1)} + \vg^{(t)}]^+ \|_2^2 \le (\alpha^{(t)})^2 \|\vr^{(t-1)} + \vg^{(t)} \|_2^2 \leq (\alpha^{(t)} )^2 ( \| \vr^{(t-1)} \|_2^2 + \| \vg^{(t)} \|_2^2 ).
    \end{equation*}
\end{proof}
A direct consequence is the following bound on the regret vector.

\begin{corollary}
    \label{cor:boundedreg}
    For any time $t \in [T]$, $\DRMplus$ guarantees
\begin{equation}
    \label{eq:unfold}
    \| \vr^{(t)} \|^2_2 \leq ( \alpha^{(t)} )^2 \|\vg^{(t)} \|_2^2 + ( \alpha^{(t)} \alpha^{(t-1)} )^2 \|\vg^{(t-1)} \|_2^2 + \dots + \left( \prod_{\tau=1}^t \alpha^{(\tau)} \right)^2 \| \vg^{(1)} \|_2^2.
\end{equation}
In particular, if $\alpha^{(t)} = 1 - \gamma$ for some constant $\gamma \in (0, 1)$, it follows that $\| \vr^{(T)} \|_2 \leq \sqrt{\nicefrac{m}{\gamma} }$.
\end{corollary}

\begin{proof}
    The first part of the claim follows by unfolding the bound of~\Cref{lemma:DRM}. For the second part, using~\eqref{eq:unfold} and the fact that $\| \vg^{(t)} \|_2^2 \leq m$ for any $t$, we have
    \begin{equation*}
        \|\vr^{(t)} \|_2^2 \leq m \left( (1-\gamma)^2 + (1 - \gamma)^4 + \dots + (1-\gamma)^{2t} \right) \leq m \frac{(1-\gamma)^2}{1 - (1 - \gamma)^2} \leq m \frac{1}{\gamma}.
    \end{equation*}
\end{proof}

\subsection{Proofs from Section~\ref{sec:RMplus}}
\label{appendix:proofs-RMplus}

We continue with the proofs from~\Cref{sec:RMplus}. We first establish that $\RMplus$ enjoys a one-step improvement property when the utility is linear.

\onestepRMplus*

\begin{proof}
First, if $\vR' = \vec{0}$, it follows that $\vR + \vu - \langle \vx, \vu \rangle \vec{1} \leq \vec{0}$, where the inequality is to be taken coordinate-wise. Since $\vR \geq \vec{0}$, we have $\langle \vx, \vu \rangle \geq \vu[a]$ for all $a \in \cA$, as claimed.

We now assume $\vR' \neq \vec{0}$. If $\vR = \vec{0}$, we have $\vR' = [ \vu - \langle \vx, \vu \rangle \vec{1}]^+$. \eqref{eq:improve-RM+} can then be equivalently expressed as
\begin{equation*}
    \sum_{a \in \cA} \vR'[a] (\vu[a] - \langle \vx, \vu \rangle) \geq \left( \max_{a \in \cA} \vu[a] - \langle \vx, \vu \rangle \right)^2,
\end{equation*}
which holds since $\vR' = [ \vu - \langle \vx, \vu \rangle \vec{1}]^+$. So we can assume $\vR \neq \vec{0}$. We define $\vdelta \defeq \vR' - \vR$. \eqref{eq:improve-RM+} can be expressed as
\begin{equation*}
    \frac{ \sum_{a \in \cA} ( \vR[a] + \vdelta[a]) \vu[a]}{ \sum_{a' \in \cA} ( \vR[a'] + \vdelta[a'] ) } \geq \frac{ \sum_{a \in \cA} \vR[a] \vu[a]}{ \sum_{a' \in \cA} \vR[a'] } + \frac{(\max_{a \in \cA} \vu[a] - \langle \vx, \vu \rangle)^2}{ \sum_{a' \in \cA} ( \vR[a'] + \vdelta[a']) }.
\end{equation*}
Equivalently,
\begin{align*}
    \sum_{a' \in \cA} \vR[a'] \sum_{a \in \cA} ( \vR[a] + \vdelta[a]) \vu[a] &\geq \sum_{a \in \cA} \vR[a] \sum_{a' \in \cA} ( \vR[a'] + \vdelta[a']) \vu[a] \\
    &+ \sum_{a' \in \cA} \vR[a'] \left( \max_{a \in \cA} \vu[a] - \langle \vx, \vu \rangle \right)^2.
\end{align*}
This in turn is equivalent to
\begin{align*}
    \sum_{a' \in \cA} \vR[a'] \sum_{a \in \cA} \vdelta[a] \vu[a] &\geq \sum_{a \in \cA} \vR[a] \sum_{a' \in \cA} \vdelta[a'] \vu[a] + \sum_{a' \in \cA} \vR[a'] \left( \max_{a \in \cA} \vu[a] - \langle \vx, \vu \rangle \right)^2 \\
    &= \sum_{a' \in \cA} \vdelta[a'] \sum_{a \in \cA} \vR[a] \langle \vx, \vu \rangle + \sum_{a' \in \cA} \vR[a'] \left( \max_{a \in \cA} \vu[a] - \langle \vx, \vu \rangle \right)^2.
\end{align*}
Rearranging,
\begin{equation*}
    \sum_{a' \in \cA} \vR[a'] \sum_{a \in \cA} \vdelta[a] ( \vu[a] - \langle \vx, \vu \rangle ) \geq \sum_{a' \in \cA} \vR[a'] \left( \max_{a \in \cA} \vu[a] - \langle \vx, \vu \rangle \right)^2.
\end{equation*}
Now, for any $a \in \cA$ such that $\vu[a] - \langle \vx, \vu \rangle \geq 0$, it follows that $\vdelta[a] = \vu[a] - \langle \vx, \vu \rangle \geq 0 $; on the other hand, for $a \in \cA$ such that $\vu[a] - \langle \vx, \vu \rangle < 0$, we have $\vdelta[a] \leq 0$. That is, $\vdelta[a] (\vu[a] - \langle \vx, \vu \rangle) \geq 0$, and the claim follows.
\end{proof}

We will now show that~\Cref{lemma:onestepRM+} is, in a certain sense, tight. We consider a simple linear maximization over the simplex. If the regret vector of $\RMplus$ can be initialized arbitrarily, as is the premise in~\Cref{lemma:onestepRM+}, we make the following observation.

\begin{lemma}
    \label{lemma:lowerRM+}
    Consider a utility vector $\vec{u} \in \R^\cA$ and some initial regret vector $\R^\cA_{\geq 0} \ni \vr^{(1)} \neq \vec{0}$. If $\vx^{(1)} = \nicefrac{\vr^{(1)}}{\| \vr^{(1)} \|_1} $ and $\epsilon = \max_{a \in \cA} \vu[a] - \langle \vx^{(1)}, \vu \rangle $ is the initial best-response gap, it takes at least $\nicefrac{\|\vr^{(1)} \|_1}{2  \epsilon }$ iterations for $\RMplus$ to reach a point $\vx^{(t)}$ with best-response gap at most $\epsilon/2$.
\end{lemma}

Indeed, we consider the two-dimensional problem in which $\vu = (1 - \epsilon, 1)$ and $\vr^{(1)} = (\|\vr^{(1)} \|_1, 0)$. To incur a best-response gap of at most $\epsilon/2$, the player needs to allot a probability mass of at least $1/2$ to the second action. In the meantime, the decrement of the first coordinate of $\vr^{(t)}$ will be at most $\epsilon$ while the increment of the second coordinate of $\vr^{(t)}$ will be at most $\epsilon$. But, for the algorithm to terminate, it must be the case that the second coordinate of $\vr^{(t)}$ is at least as large as the first coordinate of $\vr^{(t)}$, leading to~\Cref{lemma:lowerRM+}. 

Now, given that $\langle \vx^{(t)} - \vx^{(1)}, \vu \rangle \leq \epsilon$, \Cref{lemma:lowerRM+} matches the bound obtained for this problem through~\Cref{lemma:onestepRM+} in the regime where $\| \vr^{(1)} \|_1$ is at least as large as $1/\epsilon$ (so that the norm of $\vr^{(t)}$ is within a constant factor of that of $\vr^{(1)}$, by~\Cref{prop:regret bound}). 

Taking this argument a step further, if we have a regret bound of the form $\| \vr^{(t)} \|_1 \leq \|\vr\|_1 = O_t(1)$, which holds when the utility is fixed, \Cref{lemma:onestepRM+} implies that $ 2 \|\vr \|_1 + 4 \|\vr \|_1 + \dots + \|\vr \|_1/\epsilon = O_\epsilon(  1/\epsilon)$ iterations suffice for $\RMplus$ to have a best-response gap at most $\epsilon$ when facing a fixed utility; this follows by applying~\Cref{lemma:onestepRM+} first for all iterations in which the best-response gap is at least $1/2$, then for all iterations in which it is at least $1/4$, and so forth.

\paragraph{Analysis of non-lazy alternating $\RMplus$} In the main body, we used~\Cref{lemma:onestepRM+} to argue that \emph{lazy} alternating $\RMplus$ converges in potential games (\Cref{theorem:altern-RM+}). One caveat of the lazy version of alternation is that it requires knowing the desired precision $\epsilon$ ahead of time. We will now extend the analysis to encompass the usual version of alternation, albeit at the cost of introducing a further dependence on the iteration bound.

To begin with, we first state a direct refinement of~\Cref{lemma:onestepRM+} that we rely on; this refinement will also be useful in the more general setting of constrained optimization.

\begin{lemma}[Refinement of~\Cref{lemma:onestepRM+}]
    \label{lemma:refinement}
    Under the preconditions of~\Cref{lemma:onestepRM+},
    \begin{equation}
        \label{eq:refined}
        \langle \vx' - \vx, \vu \rangle \geq \frac{1}{\| \vR' \|_1} \|\vR - \vR' \|_2^2.
    \end{equation}
\end{lemma}
In particular, \eqref{eq:refined} implies~\eqref{eq:improve-RM+} since $\|\vr - \vr' \|^2_2 \geq \| \vr - \vr' \|_\infty^2 \geq ( \max_{a \in \cA} \vu[a] - \langle \vx, \vu \rangle )^2$, by definition of $\vR'$. The proof of~\Cref{lemma:refinement} is identical to that of~\Cref{lemma:onestepRM+}.

The next elementary lemma shows that, so long as the norm of the regret vector is not too small, closeness in regrets implies closeness in strategies.

\begin{lemma}
    \label{lemma:closeness}
    For $\R^\cA_{\geq 0} \ni \vr, \vr' \neq \vec{0}$, let $\vx \defeq \nicefrac{\vr}{ \|\vr\|_1}$ and $\vx' \defeq \nicefrac{\vr'}{ \|\vr'\|_1}$. Then
    \begin{equation*}
    \|\vx - \vx' \|_1 \leq \|\vR - \vR'\|_1 \left( \frac{1}{ \|\vR\|_1} + \frac{1}{\|\vR' \|_1} \right).
\end{equation*}
\end{lemma}

\begin{proof}
    The term $\vx[a] - \vx'[a]$ can be expressed, for any $a \in \cA$, as 
\begin{align*}
    \frac{ \vR[a] }{ \sum_{a' \in \cA} \vR[a']} - \frac{ \vR'[a] }{ \sum_{a' \in \cA} \vR'[a']}  &= \frac{ \sum_{a' \in \cA} (\vR[a] \vR'[a'] - \vR'[a] \vR[a']) }{ \|\vR\|_1 \|\vR'\|_1 } \\
    &= \frac{ \sum_{a' \in \cA} ( \vR[a] (\vR'[a'] - \vR[a']) + \vR[a'](\vR[a] - \vR'[a]) )  }{ \|\vR\|_1 \|\vR'\|_1 },
\end{align*}
and the claim follows.
\end{proof}

With those two helper lemmas in hand, we are now ready to analyze (non-lazy) alternating $\RMplus$ in potential games.

\begin{theorem}
    \label{theorem:nonlazy}
    In any potential game with utilities in $[0, 1]$, alternating $\RMplus$ requires at most $2 \lceil \frac{ 625 m^4 n^4 \Phimax^2 }{\delta^4 \epsilon^4} \rceil$ rounds to converge to an $\epsilon$-Nash equilibrium, where $\delta_i \defeq \BRgap_i(\vx_i^{(1)}, \vu_i^{(1)}) > 0$ and $\delta \defeq \min_{1 \leq i \leq n} \delta_i$.
\end{theorem}

The assumption that $\delta_i > 0$ is without any loss in the following sense. The analysis requires that $\| \vr^{(t)}_i \|_2 > 0$ for all $i \in [n]$ and any sufficiently large $t$. By~\Cref{prop:increasing-regret}, it suffices if a player incurs a positive best-response gap \emph{at some} time. In the contrary case, if a player always has zero best-response gap, it will always play the same strategy (by definition of $\RMplus$ in~\Cref{alg:regret_matching_plus}), thereby reducing to a potential game with $n-1$ players. We further remark that, in accordance with~\Cref{theorem:altern-RM+}, the bound in~\Cref{theorem:nonlazy} can be similarly parameterized in terms of the maximum regret incurred by a player. A more pedantic point about~\Cref{theorem:nonlazy} is that utilities are taken to be in $[0, 1]$, while in the rest of the paper we only assume that the range is bounded by $1$; this innocuous assumption is used in~\Cref{claim:smoothness} below.

\begin{proof}[Proof of~\Cref{theorem:nonlazy}]
    By~\Cref{lemma:refinement}, the telescopic summation after $T$ rounds yields
    \begin{equation}
        \label{eq:patlength}
        \Phimax \geq \sum_{t=1}^{T} \sum_{i=1}^n \frac{1}{ \| \vr_i^{(t)} \|_1 } \|\vr_i^{(t)} - \vr_i^{(t-1)} \|_2^2 \geq \sum_{t=1}^{T} \frac{1}{ \max_{i} \| \vr_i^{(t)} \|_1 } \sum_{i=1}^n \|\vr_i^{(t)} - \vr_i^{(t-1)} \|_2^2.
    \end{equation}
    Since $\|\vr_i^{(t)} \|_1 \leq m \sqrt{t}$ for any player $i \in [n]$, it follows that $\sum_{t=1}^T \sum_{i=1}^n \|\vr_i^{(t+1)} - \vr_i^{(t)} \|_2^2 \leq m \Phimax \sqrt{T}$. As a result, after at most $2 \lceil \nicefrac{m^2 \Phimax^2}{\epsilon^4} \rceil$ rounds, there will be a time $t$ such that $\sum_{i=1}^n \| \vr_i^{(t)} - \vr_i^{(t-1)} \|_2^2 + \sum_{i=1}^n \| \vr_i^{(t+1)} - \vr_i^{(t)} \|_2^2 \leq \epsilon^2$, which in turn implies $\| \vr_i^{(t)} - \vr_i^{(t-1)} \|_2, \| \vr_i^{(t+1)} - \vr_i^{(t)} \|_2 \leq \epsilon$ for all $i \in [n]$. Now, by assumption, we know that $\| \vr^{(1)}_i \|_2 \geq \delta_i > 0$ for all $i \in [n]$. By the monotonicity property of the regret vector (\Cref{prop:increasing-regret}, proven in the sequel), it follows that $\| \vr^{(t)}_i \|_2 \geq \delta_i > 0$ for all $t \in [T+1]$ and $i \in [n]$. Consequently, applying~\Cref{lemma:closeness}, we have $\| \vx_i^{(t+1)} - \vx_i^{(t)} \|_1 \leq 2 \| \vr_i^{(t+1)} - \vr_i^{(t)} \|_1/\delta_i \leq 2 \sqrt{m} \epsilon/\delta_i$ since $\|\cdot\|_1 \leq \sqrt{m} \|\cdot\|_2$. Next, we will make use of the following simple claim.

    \begin{claim}
        \label{claim:smoothness}
        Consider a normal-form game with utilities in $[0, 1]$. For any two joint strategies $(\vx_1, \dots, \vx_n)$ and $(\vx_1', \dots, \vx_n')$, it holds that $\| \vu_i(\vx_{-i}) - \vu_i(\vx_{-i}') \|_\infty \leq \sum_{i' \neq i} \|\vx_{i'} - \vx_{i'}'\|_1$ for any player $i \in [n]$.
    \end{claim}

    \begin{proof}
        We fix a player $i \in [n]$. We have
        \begin{align}
            \| \vu_i(\vx_{-i}) - \vu_i(\vx_{-i}') \|_\infty &= \left| \sum_{\vec{a} \in \cA_1 \times \dots \times \cA_n} u_i(\cdot, \vec{a}_{-i}) \left( \prod_{i' \neq i} \vx_{i'}[a_{i'}] - \prod_{i' \neq i} \vx'_{i'}[a_{i'}] \right) \right| \notag \\
            &\leq \left| \sum_{\vec{a} \in \cA_1 \times \dots \times \cA_n} \left( \prod_{i' \neq i} \vx_{i'}[a_{i'}] - \prod_{i' \neq i} \vx'_{i'}[a_{i'}] \right) \right| \label{align:first-TV} \\
            &\leq \sum_{i' \neq i} \|\vx_{i'} - \vx_{i'}'\|_1,\label{align:second-TV}
        \end{align}
        where~\eqref{align:first-TV} uses triangle inequality together with the assumption that $|u_i(\cdot)| \leq 1$, and~\eqref{align:second-TV} uses a bound on the total variation distance of a product distribution in terms of the sum of the total variation distances of its marginals~\citep{Hoeffding58:Distinguishability}.
    \end{proof}
    Using this lemma, we now observe that $\| \vu_i^{(t)} - \vu_i(\vx_{-i}^{(t)}) \|_\infty \leq \sum_{i' < i} \|\vx^{(t+1)}_{i'} - \vx^{(t)}_{i'}\|_1 \leq \sum_{i' \neq i} \|\vx^{(t+1)}_{i'} - \vx^{(t)}_{i'}\|_1 \leq 2 n \sqrt{m} \epsilon / \delta$ for each $i \in [n]$, where $\delta = \min_{1 \leq i \leq n} \delta_i$. Furthermore, in view of the fact that $\| \vr_i^{(t)} - \vr_i^{(t-1)} \|_2 \leq \epsilon$, it follows that $\| \vu^{(t)}_i \|_\infty - \langle \vx_i^{(t)}, \vu_i^{(t)} \rangle \leq \epsilon$. Thus,
    \begin{align}
        \| \vu_i(\vx^{(t)}_{-i}) \|_\infty - \langle \vx_i^{(t)}, \vu_i(\vx^{(t)}_{-i}) \rangle &=  \| \vu^{(t)}_i \|_\infty - \langle \vx_i^{(t)}, \vu_i^{(t)} \rangle + \| \vu_i(\vx^{(t)}_{-i}) \|_\infty - \| \vu^{(t)}_i \|_\infty + \langle \vx_i^{(t)}, \vu^{(t)}_i - \vu_i(\vx^{(t)}_{-i}) \rangle \notag \\
        &\leq \| \vu^{(t)}_i \|_\infty - \langle \vx_i^{(t)}, \vu_i^{(t)} \rangle + 2 \| \vu_i^{(t)} - \vu_i(\vx_{-i}^{(t)}) \|_\infty \label{align:first-udiff} \\
        &\leq \epsilon \left(1 + \frac{4 n \sqrt{m}}{\delta}\right) \leq \epsilon \left( 5 \frac{n \sqrt{m}}{\delta} \right). \notag
    \end{align}
    where~\eqref{align:first-udiff} follows from the fact that $\langle \vx_i^{(t)}, \vu^{(t)}_i - \vu_i(\vx^{(t)}_{-i}) \rangle \leq \| \vx_i^{(t)} \|_1 \| \vu^{(t)}_i - \vu_i(\vx^{(t)}_{-i}) \|_\infty \leq \| \vu^{(t)}_i - \vu_i(\vx^{(t)}_{-i}) \|_\infty$ since $\| \vx_i^{(t)} \|_1 = 1$. We conclude that $(\vx_1^{(t)}, \dots, \vx_n^{(t)})$ is an $\epsilon ( \nicefrac{5 n \sqrt{m}}{\delta})$-Nash equilibrium; rescaling $\epsilon$ leads to the claim.
\end{proof}

\begin{remark}[Convergence to CCE under alternating updates]
    \label{remark:alternating-CCE}
    The folk connection linking no-regret learning and coarse correlated equilibria (\Cref{prop:CCE}) is predicated on the dynamics being executed in a simultaneous fashion. We observe that, in potential games, even alternating $\RMplus$ converges to the set of CCEs in the following sense. \Cref{lemma:closeness} together with~\eqref{eq:patlength} imply that $\sum_{t=1}^T \sum_{i=1}^n \| \vx_i^{(t+1)} - \vx_i^{(t)} \|^2_1 = O_T(\sqrt{T})$. The Cauchy-Schwarz inequality in turn yields 
    \begin{equation}
        \label{eq:fo-pathlength}
      \sum_{t=1}^T \sum_{i=1}^n \| \vx_i^{(t+1)} - \vx_i^{(t)} \|_1 \leq \sqrt{\sum_{t=1}^T \left( \sum_{i=1}^n \| \vx_i^{(t+1)} - \vx_i^{(t)} \|_1 \right)^2 \sum_{t=1}^T 1^2 } = O_T(T^{3/4}).
    \end{equation}
    Now, by the fact that $\RMplus$ has the no-regret property (\Cref{prop:noregret-RM}), $\max_{\vx'_i \in \Delta(\cA_i)} \sum_{t=1}^T \langle \vx_i' - \vx_i^{(t)}, \vu_i^{(t)} \rangle = O_T(\sqrt{T})$. Further, by~\Cref{claim:smoothness} and~\eqref{eq:fo-pathlength}, $\sum_{t=1}^T \|\vu_i(\vx^{(t)}_{-i}) - \vu_i^{(t)} \|_\infty \leq \sum_{t=1}^T \sum_{i' \neq i} \| \vx_{i'}^{(t+1)} - \vx_{i'}^{(t)} \|_1 = O_T(T^{3/4})$. As a result, we conclude that, for any player $i \in [n]$,
    \begin{equation*}
        \max_{\vx'_i \in \Delta(\cA_i)} \sum_{t=1}^T \langle \vx_i' - \vx_i^{(t)}, \vu_i(\vx_{-i}^{(t)}) \rangle = O_T(T^{3/4}).
    \end{equation*}
    This implies that the correlated distribution $\frac{1}{T} \sum_{t=1}^T \bigotimes_{i=1}^n \vx_i^{(t)}$ is an $\epsilon$-CCE for some $\epsilon = O_T(T^{-1/4})$ (\Cref{prop:CCE}); here, $\bigotimes_{i=1}^n \vx_i^{(t)}$ denotes the product distribution induced by $(\vx_1^{(t)}, \dots, \vx_n^{(t)})$.
\end{remark}

Moving on, a further implication of~\Cref{lemma:lowerRM+} is that having a large regret vector can slow down $\RMplus$. Employing discounting, in the form of $\DRMplus$ (\Cref{alg:regret_matching_plus-discounting}), addresses this deficiency in potential games, as we prove below. It is worth stressing that while discounting is sensible when employing alternating $\RMplus$ in potential games, this is not the case in more general constrained optimization problems; there, having a regret vector with a \emph{small} norm can impede convergence (\emph{cf.}~\Cref{lemma:R-L}).

\DRMcor*

\begin{proof}
    \Cref{lemma:onestepRM+} can be directly adjusted for $\DRMplus$: the updated strategy $\vx' = \nicefrac{\vr'}{\| \vr' \|_1 } $, where $\vR' = \alpha [ \vR + \vu - \langle \vx, \vu \rangle \vec{1}]^+ \neq \vec{0}$, remains the same since we only rescaled the regret vector by $\alpha$. That is, for any $\vu \in \R^{\cA}$,
    \begin{equation*}
        \langle \vx' - \vx, \vu \rangle \geq \frac{1}{\| \vR' \|_1} \left( \max_{a \in \cA} \vu[a] - \langle \vx, \vu \rangle \right)^2.
    \end{equation*}
    The key difference compared to~\Cref{theorem:altern-RM+} is that we now maintain the invariance $\| \vr_i^{(t)} \|_2 \leq \sqrt{\nicefrac{m}{\gamma}}$ for all $i \in [n]$ and $t \in [T]$ (\Cref{cor:boundedreg}). Consequently, 
    \begin{equation*}
        \Phimax \geq \sum_{t=1}^T \sum_{i=1}^n \frac{1}{ \| \vr_i^{(t)} \|_1 } \BRgap_i(\vx_i^{(t)}, \vu_i^{(t)})^2 \mathbbm{1} \{ \BRgap_i(\vx_i^{(t)}, \vu_i^{(t)} ) > \epsilon \}.
    \end{equation*}
    If in every round $t \in [T]$ there is a player $i \in [n]$ such that $\BRgap_i(\vx_i^{(t)}, \vu_i^{(t)}) > \epsilon$, we have $\Phimax \geq \sum_{t=1}^T (\nicefrac{\sqrt{\gamma}}{m}) \epsilon^2$, so $T \leq \frac{m \Phimax}{\epsilon^2 \sqrt{\gamma}}$. This concludes the proof.
\end{proof}

Unlike $\RMplus$ and $\DRMplus$, $\RM$ only has a \emph{conditional} one-step improvement because the regret vector can have negative coordinates.

\onestepRM*

\begin{proof}
    We define $\vdelta \defeq \vtheta' - \vtheta$. Following the proof of~\Cref{lemma:onestepRM+}, it suffices to show that
    \begin{equation}
        \label{eq:RM-goal}
    \sum_{a' \in \cA} \vtheta[a'] \sum_{a \in \cA} \vdelta[a] ( \vu[a] - \langle \vx, \vu \rangle ) \geq \sum_{a' \in \cA} \vtheta[a'] \left( \max_{a \in \cA} \vu[a] - \langle \vx, \vu \rangle \right)^2 \mathbbm{1} \left\{ \vR[a] \geq 0 \right\}.
\end{equation}
For an action $a \in \cA$, we consider the following cases.
\begin{itemize}
    \item If $\vu[a] - \langle \vx, \vu \rangle \geq 0$,
    \begin{itemize}
        \item if $\vR[a] \geq 0$, we have $\vdelta[a] = \vu[a] - \langle \vx, \vu \rangle$.
        \item If $\vR[a] < 0$, it follows that $\vdelta[a] \geq 0$; in particular, $\vdelta[a] = 0$ if $\vR'[a] \leq 0$ and $\vdelta[a] > 0$ otherwise. As a result, $\vdelta[a] ( \vu[a] - \langle \vx, \vu \rangle ) \geq 0$.
    \end{itemize}
    \item If $\vu[a] - \langle \vx, \vu \rangle < 0$,
    \begin{itemize}
        \item if $\vR[a] \leq 0$, we have $\vdelta[a] = 0$ since $\vtheta[a] = 0 = \vtheta'[a]$.
        \item if $\vR[a] > 0$, it follows that $\vdelta[a] < 0$. Again, we have $\vdelta[a] ( \vu[a] - \langle \vx, \vu \rangle ) \geq 0$.
    \end{itemize}
\end{itemize}
Combining those items, \eqref{eq:RM-goal} follows.
\end{proof}

We now turn to the more general setting of constrained optimization. First, combining~\Cref{lemma:closeness,lemma:refinement} that were proven earlier, we formally show that $\RMplus$ improves the value of the underlying function provided that the norm of the regret vector is not too small.

\condiRMplus*

\begin{proof}
    Using the quadratic bound for $u$, we have
    \begin{align}
        u(\vx') - u(\vx) &\geq  \langle \nabla u(\vx), \vx' - \vx \rangle - \frac{L}{2} \|\vx - \vx'\|_2^2 \notag \\
        &\geq \frac{1}{ \|\vr' \|_1} \|\vr - \vr' \|_2^2 - \frac{L}{2} \|\vr - \vr' \|_1^2 \left( \frac{1}{\|\vr\|_1} + \frac{1}{\|\vr'\|_1}  \right)^2 \label{align:firstb} \\
        &\geq \frac{1}{ \|\vr' \|_1} \|\vr - \vr' \|_2^2 - \frac{9 m L}{2 \|\vr'\|_1^2} \|\vr - \vr' \|_2^2 \label{align:secondb} \\
        &\geq \frac{1}{2 \|\vr'\|_1} \|\vr - \vr' \|_2^2, \label{align:thirdb}
    \end{align}
    where~\eqref{align:firstb} uses the one-step improvement property (\Cref{lemma:refinement}) applied for $\vu \defeq \nabla u(\vx)$ together with~\Cref{lemma:closeness}; \eqref{align:secondb} follows from the fact that $\|\vR\|_1 \geq \|\vR'\|_1 - m \geq \frac{1}{2} \|\vR'\|_1$ since $\|\vr'\|_1 \geq \|\vr'\|_2 \geq 2 m$ and $|\langle \vx - \vx', \nabla u (\vx) \rangle | \leq 1$ for all $\vx' \in \Delta(\cA)$ (per our normalization assumption); and~\eqref{align:thirdb} follows from the assumption that $\| \vr' \|_1 \geq 9 m L$.
\end{proof}

To make use of~\Cref{lemma:R-L}, we next establish that the $\ell_2$ norm of the regret vector under $\RMplus$ is nondecreasing.

\nondecreasingregret*

\begin{proof}
We have $\vr^{(t)} - \vr^{(t-1)} = \max(\vg^{(t)}, -\vr^{(t-1)})$ (element-wise), so
\begin{equation*}
  \| \vr^{(t)} - \vr^{(t-1)} \|_2 = \| \max(\vg^{(t)}, -\vr^{(t-1)}) \|_2 \ge \| [\vg^{(t)}]^+ \|_2.
\end{equation*}

Further, $\langle \vr^{(t-1)}, \vr^{(t)} - \vr^{(t-1)} \rangle = \langle \vr^{(t-1)}, \max(\vg^{(t)}, -\vr^{(t-1)}) \rangle \ge \langle \vr^{(t-1)}, \vg^{(t)} \rangle = 0$, where we used the fact that $\vr^t \ge \vec 0$, coordinate-wise. Therefore,
    \begin{align*}
    \| \vr^{(t)} \|_2^2 = \| \vr^{(t)} - \vr^{(t-1)} + \vr^{(t-1)} \|_2^2 &= \| \vr^{(t)} - \vr^{(t-1)} \|_2^2 + \| \vr^{(t-1)} \|_2^2 + 2 \langle \vr^{(t-1)}, \vr^{(t)} - \vr^{(t-1)} \rangle \\
    &\geq \| \vr^{(t-1)} \|_2^2 + \| [\vg^{(t)}]^+ \|_2^2,
    \end{align*}
    as claimed.
\end{proof}

Armed with~\Cref{lemma:R-L,prop:increasing-regret}, we can now prove~\Cref{theorem:singlesimplex}.

\singlesimplex*

\begin{proof}
    Let $\tcrit \in [T]$ be the largest $t$ such that $\| \vr^{(t)} \|_2 < \max \{2 m, 9 m L \} = R$. By~\Cref{prop:increasing-regret},
    \begin{equation*}
    \| \vr^{(t)} \|_2^2 \ge \| \vr^{(t-1)} \|_2^2 + \| [\vg^{(t)}]_+ \|_2^2 \geq \| \vr^{(t-1)} \|_2^2 +  \KKTgap(\vx^{(t)})^2;
\end{equation*}
so,
    \begin{equation}
        \label{eq:before}
        \sum_{t=1}^{\tcrit}  \KKTgap(\vx^{(t)})^2 \leq \sum_{t=1}^{\tcrit} ( \|\vr^{(t)} \|_2^2 - \|\vr^{(t-1)} \|_2^2) = \| \vr^{(\tcrit)} \|_2^2  \leq R^2.
    \end{equation}
    Further, for any $t \geq t_c + 1$, we have $\| \vr^{(t)} \|_2 \geq R$ since the $\ell_2$ norm of the regret vector is nondecreasing (\Cref{prop:increasing-regret}) and $\| \vr^{(\tcrit+1)} \|_2 \geq R$. Thus, by~\Cref{lemma:R-L},
    \begin{equation}
        \label{eq:after}
        \sum_{t=t_c+1}^T \frac{1}{2 \|\vr^{(t)} \|_1 }  \KKTgap(\vx^{(t)})^2 \leq u(\vx^{(T+1)}) - u(\vx^{(\tcrit+1)}) \leq \urange.
    \end{equation}
    Combining~\eqref{eq:before} and~\eqref{eq:after}, together with the fact that $\| \vr^{(t)} \|_1 \leq m \sqrt{t}$,
    \begin{equation*}
        \sum_{t=1}^T \frac{1}{ m \sqrt{t} } \KKTgap(\vx^{(t)})^2 \leq 2 \urange + R^2.
    \end{equation*}
    This leads to the claim.
\end{proof}

Next, we use~\Cref{theorem:singlesimplex} to establish that even \emph{simultaneous} $\RMplus$ converges in symmetric potential games, as long as the players have the same initialization.

\sympot*

\begin{proof}
    We will argue that simultaneous $\RMplus$, under the same initialization, in a symmetric game is tantamount to running $\RMplus$ with respect to the function $\Delta(\cA) \ni \vx \mapsto \Phi(\vx, \dots, \vx)$, where $\cA_1 = \dots = \cA_n = \cA$. We first claim that for any $a \in \cA$,
    \begin{equation}
        \label{eq:firstclaim}
        \frac{\partial \Phi(\vx, \dots, \vx) }{\partial \vx[a]} = \sum_{i=1}^n \frac{\partial \Phi(\vx_1, \dots, \vx_n) }{\partial \vx_i[a]} \bigg\mid_{\vx_1 = \dots = \vx_n = \vx}.
    \end{equation}
    By definition, $\Phi(\vx_1, \dots, \vx_n)$ is multilinear. Let us consider a monomial of $\Phi$ of the form
    \begin{equation*}
        m_{\vec{a}}(\vx_1, \dots, \vx_n) = \prod_{i = 1}^n \vx_i[a_i]
    \end{equation*}
    for some joint action $\va = (a_1, \dots, a_n) \in \cA_1 \times \dots \times \cA_n$. Then
    \begin{align}
        \frac{\partial m_{\vec{a}}(\vx, \dots, \vx)}{\partial \vx[a]} &= \frac{\partial}{\partial \vx[a]} \left( \prod_{a' \in \cA'} (\vx[a'])^{d(a')} \right) \notag \\
        &=
        \begin{cases}
            d(a) (\vx[a])^{d(a) - 1} \prod_{a' \in \cA' \setminus \{ a \}} (\vx[a'])^{d(a')} & \text{if } a \in \cA' \text{ and} \\
            0 & \text{otherwise},
        \end{cases} \label{align:cases}
    \end{align}
    where $\cA' = \cA'(\vec{a}) = \{ a' \in \cA : \exists i \in [n] \text{ such that } a' = a_i \} $ and degrees $d(a') = | \{ i \in [n] : a_i = a' \} | \geq 1 $ for $a' \in \cA'$. Further,
    \begin{equation*}
        \frac{ \partial m_{\vec{a}}(\vx_1, \dots, \vx_n)}{\partial \vx_i[a]} = \begin{cases}
            \prod_{i' \neq i} \vx_{i'}[a_{i'}] & \text{if } a = a_i \\
            0 & \text{otherwise}.
        \end{cases}
    \end{equation*}
    In particular,
    \begin{equation*}
    \frac{ \partial m_{\vec{a}}(\vx_1, \dots, \vx_n)}{\partial \vx_i[a]} \bigg\mid_{\vx_1 = \dots = \vx_n} = \mathbbm{1} \{ a = a_i \} \vx[a]^{d(a) - 1} \prod_{a' \in \cA' \setminus \{a \}} (\vx[a'])^{d(a')}.
    \end{equation*}
    As a result,
    \begin{equation*}
        \sum_{i=1}^n \frac{ \partial m_{\vec{a}}(\vx_1, \dots, \vx_n)}{\partial \vx_i[a]} \bigg\mid_{\vx_1 = \dots = \vx_n} = \mathbbm{1} \{ a \in \cA' \} d(a) \vx[a]^{d(a) - 1} \prod_{a' \in \cA' \setminus \{a \}} (\vx[a'])^{d(a')},
    \end{equation*}
    which matches the expression in~\eqref{align:cases}. That is, we have shown that for any $\vec{a} \in \cA_1 \times \dots \times \cA_n$,
    \begin{equation*}
        \frac{\partial m_{\vec{a}}(\vx, \dots, \vx)}{\partial \vx[a]} = \sum_{i=1}^n \frac{ \partial m_{\vec{a}}(\vx_1, \dots, \vx_n)}{\partial \vx_i[a]} \bigg\mid_{\vx_1 = \dots = \vx_n}.
    \end{equation*}
    Since $\Phi(\vx_1, \dots, \vx_n) = \sum_{\vec{a} \in \cA_1 \times \dots \times \cA_n} \Phi(\vec{a}) m_{\vec{a}}(\vx_1, \dots, \vx_n)$, \eqref{eq:firstclaim} follows by linearity. Moreover, the symmetry assumption concerning the potential $\Phi$ tells us that
    \begin{equation*}
        \frac{ \partial \Phi(\vx_1, \dots, \vx_n)}{ \partial \vx_i[a]} \bigg\mid_{\vx_1 = \dots = \vx_n} = \frac{ \partial \Phi(\vx_1, \dots, \vx_n)}{ \partial \vx_{i'}[a]} \bigg\mid_{\vx_1 = \dots = \vx_n}
    \end{equation*}
    for any $i, i' \in [n]$. Combining with~\eqref{eq:firstclaim}, we have that for any $a \in \cA$,
    \begin{equation*}
        \frac{\partial \Phi(\vx, \dots, \vx) }{\partial \vx[a]} = n \frac{\partial \Phi(\vx_1, \dots, \vx_n) }{\partial \vx_i[a]} \bigg\mid_{\vx_1 = \dots = \vx_n = \vx} \forall i \in [n].
    \end{equation*}
    Given that $\RMplus$ is scale invariant, we conclude that simultaneous $\RMplus$ on the potential game is equivalent to running $\RMplus$ on the function $ \Delta(\cA) \ni \vx \mapsto \Phi(\vx, \dots, \vx)$, and the claim follows from~\Cref{theorem:singlesimplex}.
\end{proof}

We now turn to the analysis of alternating $\RMplus$ in constrained optimization problems over multiple probability simplices. We clarify that when the regret vector is initialized at $\R_{\geq 0}^{\cA_i} \ni \vr_i^{(0)} \neq \vec{0}$, as assumed below, the first strategy must be defined consistently as $\vx_i^{(1)} \propto \vr_i^{(0)}$.

\easycor*

\begin{proof}
    By~\Cref{prop:increasing-regret}, it follows that $\| \vr_i^{(t)} \|_2 \geq \| \vr_i^{(0)} \|_2 = \max \{ 2 m_i, 9 m_i L \} $. Following the proof of~\Cref{lemma:R-L}, we have
    \begin{equation*}
        u(\vx^{(t+1)}_{i' \leq i}, \vx^{(t)}_{i' > i} ) - u(\vx_{i' < i}^{(t+1)}, \vx_{i' \geq i}^{(t)}) \geq \frac{1}{2 \|\vr_i^{(t)}\|_1} \left( \BRgap_i(\vx_i^{(t)}, \vu_i^{(t)} )  \right)^2 \mathbbm{1} \{ \BRgap_i(\vx_i^{(t)}, \vu_i^{(t)} ) > \epsilon \}.
    \end{equation*}
    As a result, the telescopic summation yields
    \begin{equation*}
        \sum_{t=1}^T \sum_{i=1}^n \frac{1}{2 \|\vr_i^{(t)}\|_1} \left( \BRgap_i(\vx_i^{(t)}, \vu_i^{(t)} )  \right)^2 \mathbbm{1} \{ \BRgap_i(\vx_i^{(t)}, \vu_i^{(t)} ) > \epsilon \} \leq \urange,
    \end{equation*}
    Since $\|\vr_i^{(t)} \|_1 \leq m \sqrt{t}$ for all $i \in [n]$ and $t \in [T]$, it will take at most $1 + \nicefrac{ 4 m^2 \urange^2}{\epsilon^4}$ rounds to converge to a point in which all players have at most an $\epsilon$ best-response gap, which in turn implies that the KKT gap is at most $n \epsilon$. Rescaling $\epsilon$ concludes the proof.
\end{proof}

To analyze simultaneous $\RMplus$, we adapt~\Cref{lemma:R-L} as follows.

\begin{lemma}
    \label{lemma:R-L-multi}
    Let $u$ be an $L$-smooth function over $\Delta(\cA_1) \times \dots \times \Delta(\cA_n)$. For any $i \in [n]$ and $\vr_i \in \R^{\cA_i}_{\geq 0}$ with $\vr_i \neq \vec{0}$, we define $\vx_i \defeq \nicefrac{\vr_i}{\|\vr_i\|_1}$. Further, let $\vr_i' \defeq [ \vr_i + \nabla_{\vx_i} u(\vx) - \langle \vx_i, \nabla_{\vx_i} u(\vx) \rangle \vec{1} ]^+ \neq \vec{0}$ and $\vx_i' \defeq \nicefrac{\vr_i'}{ \|\vr_i'\|_1}$. If $\|\vr_i'\|_2 \geq \max\{ 2 m_i, 9 m_i L \}$ for any $i \in [n]$, then
    \begin{equation*}
        u(\vx') - u(\vx) \geq \frac{1}{2 \max_{1 \leq i \leq n} \|\vr_i'\|_1} \sum_{i=1}^n \left( \max_{\vxstar_i \in \Delta(\cA_i)} \langle \vxstar_i - \vx_i, \nabla_{\vx_i} u(\vx)  \rangle  \right)^2.
    \end{equation*}
\end{lemma}

\begin{proof}
    Using the quadratic bound for $u$, we have
    \begin{align}
        u(\vx') - u(\vx) &\geq  \langle \nabla u(\vx), \vx' - \vx \rangle - \frac{L}{2} \|\vx - \vx'\|_2^2 \notag \\
        &= \sum_{i=1}^n \left( \langle \nabla_{\vx_i} u(\vx), \vx_i' - \vx_i \rangle - \frac{L}{2} \|\vx_i - \vx_i'\|_2^2 \right) \notag \\
        &\geq \sum_{i=1}^n \left( \frac{1}{ \|\vr_i' \|_1} \|\vr_i - \vr_i' \|_2^2 - \frac{L}{2} \|\vr_i - \vr_i' \|_1^2 \left( \frac{1}{\|\vr_i\|_1} + \frac{1}{\|\vr_i'\|_1}  \right)^2 \right) \label{align:firstb-new} \\
        &\geq \sum_{i=1}^n \left( \frac{1}{ \|\vr_i' \|_1} \|\vr_i - \vr_i' \|_2^2 - \frac{9 m_i L}{2 \|\vr_i'\|_1^2} \|\vr_i - \vr_i' \|_2^2 \right) \label{align:secondb-new} \\
        &\geq \frac{1}{2} \sum_{i=1}^n \left( \frac{1}{\|\vr_i'\|_1} \|\vr_i - \vr_i' \|_2^2 \right), \label{align:thirdb-new}
    \end{align}
    where~\eqref{align:firstb-new} uses the one-step improvement property (\Cref{lemma:refinement}) applied for each player $i \in [n]$ and $\vu_i \defeq \nabla_{\vx_i} u(\vx)$ together with~\Cref{lemma:closeness}; \eqref{align:secondb-new} follows from the fact that $\|\vr_i \|_1 \geq \|\vr_i'\|_1 - m_i \geq \frac{1}{2} \|\vr_i'\|_1$ since $\|\vr_i'\|_1 \geq \|\vr_i'\|_2 \geq 2 m_i$ and $|\langle \vx_i - \vx_i', \nabla_{\vx_i} u (\vx) \rangle | \leq 1$ for all $\vx_i' \in \Delta(\cA_i)$ (per our normalization assumption); and~\eqref{align:thirdb-new} follows from the assumption that $\| \vr_i' \|_1 \geq 9 m_i L$ for each $i \in [n]$.
\end{proof}

Similarly to~\Cref{cor:easycor}, we can now show the following concerning simultaneous $\RMplus$.

\begin{corollary}
    \phantomsection
    \label{cor:easycor-sim}
    If $u$ is an $L$-smooth function in $\Delta(\cA_1) \times \dots \times \Delta(\cA_n)$ with range $\urange$, simultaneous $\RMplus$ initialized at $\vr^{(0)}_i = \max\{ 2\sqrt{m_i}, 9 \sqrt{m_i} L \} \vec{1} $ for each player $i \in [n]$ requires at most $1 + \frac{ 4 n^2 m^2 \urange^2 }{\epsilon^4}$ rounds to reach an $\epsilon$-KKT point of $u$.
\end{corollary}

\begin{proof}
    By~\Cref{prop:increasing-regret}, it follows that $\| \vr_i^{(t)} \|_2 \geq \| \vr_i^{(0)} \|_2 = \max \{ 2m_i, 9 m_i L \} $. By~\Cref{lemma:R-L-multi},
    \begin{equation*}
        u(\vx^{(t+1)}) - u(\vx^{(t)}) \geq \frac{1}{2n} \frac{1}{\max_{1 \leq i \leq n} \|\vr_i^{(t)} \|_1} \left( \KKTgap(\vx^{(t)})  \right)^2,
    \end{equation*}
    and the claim follows.
\end{proof}

Next, we show how to extend the analysis even when the regret vector is initialized at zero, at the cost of incurring a worse dependence on $1/\epsilon$.

\mainth*

\begin{proof}
    We analyze $\epsilon$-lazy simultaneous $\RMplus$; the alternating version can be treated similarly. If $\|\vr_i^{(t)} \|_2 \geq \max\{ 2 m_i, 9 m_i L \}$ for all players with best-response gap at least $\epsilon$, we have
    \begin{equation}
        \label{eq:urange}
        u(\vx^{(t+1)}) - u(\vx^{(t)}) \geq \frac{1}{2 \max_{1 \leq i \leq n} \|\vr_i^{(t)}\|_1} \sum_{i=1}^n \left( \BRgap_i(\vx_i^{(t)}, \vu_i^{(t)} )  \right)^2 \mathbbm{1} \{ \BRgap_i(\vx_i^{(t)}, \vu_i^{(t)} ) > \epsilon \};
    \end{equation}
    this follows similarly to~\Cref{lemma:R-L-multi} since players with best-response gap below $\epsilon$ do not update their strategies (under lazy updates). Now, by~\Cref{prop:increasing-regret}, it follows that the total number of rounds in which there is a player $i$ with at least an $\epsilon$ best-response gap and $\| \vr_i^{(t)} \|_2 < \max\{2m_i, 9 m_i L \} = R_i$ is at most $\sum_{i=1}^n (1 + \nicefrac{ R_i^2 }{\epsilon^2})$. We will now bound the number of rounds $T'$ it takes to encounter two such rounds. We observe that, between such rounds, we only update players that have at least an $\epsilon$ best-response gap and $\| \vr_i^{(t)} \|_2 \geq R_i$. By~\eqref{eq:urange}, this can continue for at most $1 + \nicefrac{2 \urange m \sqrt{T}}{\epsilon^2}$ rounds, where $T$ is the total number of rounds it takes for all players to have a best-response gap of at most $\epsilon$. That is, $T' \leq 1 + \nicefrac{2 \urange m \sqrt{T}}{\epsilon^2}$. Further, $T \leq T' + T' \sum_{i=1}^n (1 + \nicefrac{ R_i^2 }{\epsilon^2})$, and the claim follows by solving in terms of $T$ and rescaling $\epsilon$.
\end{proof}

\subsection{Proofs from Section~\ref{sec:RM}}
\label{appendix:lowerbound}

We conclude with the proofs from~\Cref{sec:RM}. Below (left figure), we provide an illustrative example of the matrix $\mat{B}$, defined earlier in~\eqref{eq:gamelower}, for $m = 6$. The plots in~\Cref{fig:exp-sep} are obtained by running simultaneous $\RM$ (left) and alternating $\RMplus$ (right) on this exact game. It is worth pointing out that one could also use a staircase---instead of a spiral---pattern as shown in the right figure. 

\begin{center}
    \scalebox{0.8}{\begin{tikzpicture}[baseline=(current bounding box.center)]
  \matrix (m) [matrix of nodes,
               nodes in empty cells,
               nodes={minimum size=8mm, anchor=center},
               left delimiter={[},
               right delimiter={]},
               row sep=0.6em, column sep=0.6em] {
    |[name=a1]| 1 & 0 & 0 & 0 & 0 & 0 & 0.5 \\
    0 & |[name=a5]| 5 & 0 & 0 & 0 & |[name=a4]| 4 & 0 \\
    0 & 0 & |[name=a9]| 9 & 0 & |[name=a8]| 8 & 0 & 0 \\
    0 & 0 & |[name=a10]| 10 & |[name=a11]| 11 & 0 & 0 & 0 \\
    0 & |[name=a6]| 6 & 0 & 0 & |[name=a7]| 7 & 0 & 0 \\
    |[name=a2]| 2 & 0 & 0 & 0 & 0 & |[name=a3]| 3 & 0 \\
    0.5 & 0 & 0 & 0 & 0 & 0 & |[name=a0]| 0 \\
  };

  \draw[->, thin, black, loosely dotted] (a0) to (a1);
  \draw[->, thick, red] (a1) to[bend right=15] (a2);
  \draw[->, thick, red] (a2) to[bend right=15] (a3);
  \draw[->, thick, red] (a3) to[bend right=15] (a4);
  \draw[->, thick, red] (a4) to[bend right=15] (a5);
  \draw[->, thick, red] (a5) to[bend right=15] (a6);
  \draw[->, thick, red] (a6) to[bend right=15] (a7);
  \draw[->, thick, red] (a7) to[bend right=15] (a8);
    \draw[->, thick, red] (a8) to[bend right=15] (a9);
    \draw[->, thick, red] (a9) -- (a10);
    \draw[->, thick, red] (a10) -- (a11);
\end{tikzpicture}
}
    \hspace{2em}
    \scalebox{0.8}{

\begin{tikzpicture}[baseline=(current bounding box.center)]
  \matrix (m) [matrix of nodes,
               nodes in empty cells,
               nodes={minimum size=8mm, anchor=center},
               left delimiter={[},
               right delimiter={]},
               row sep=0.6em, column sep=0.6em] {
    |[name=a1]| 1 & |[name=a2]| 2 & 0 & 0 & 0 & 0\\
    0 & |[name=a3]| 3 & |[name=a4]| 4 & 0 & 0 & 0 \\
    0 & 0 & |[name=a5]| 5 & |[name=a6]| 6 & 0 & 0 \\
    0 & 0 & 0 & |[name=a7]| 7 & |[name=a8]| 8 & 0 \\
    0 & 0 & 0 & 0 & |[name=a9]| 9 & |[name=a10]| 10 \\
    0 & 0 & 0 & 0 & 0 & |[name=a11]| 11 \\
  };

  \draw[->, thick, red] (a1) to (a2);
  \draw[->, thick, red] (a2) to (a3);
  \draw[->, thick, red] (a3) to (a4);
  \draw[->, thick, red] (a4) to (a5);
  \draw[->, thick, red] (a5) to (a6);
  \draw[->, thick, red] (a6) to (a7);
  \draw[->, thick, red] (a7) to (a8);
    \draw[->, thick, red] (a8) to (a9);
    \draw[->, thick, red] (a9) -- (a10);
    \draw[->, thick, red] (a10) -- (a11);
\end{tikzpicture}
}
\end{center}

Our main goal is to prove the following invariance.

\invariance*

It is possible to check the claim for $k = 1, 2, 3, 4$ by executing $\RM$ for a certain number of rounds. In particular, we find that $T_3 \geq 5$ and $T_4 \geq 20$. We proceed by induction in $k$. Suppose that it holds for all payoffs $1, \dots, \kappa$. We will show that it holds for $\kappa + 1$.

\begin{lemma}
    \label{lemma:indone}
    For any even $\kappa+2 \geq k \geq 4$, let $\vr_1^{(\thigh{k-2})}[a_1]$ be the regret of Player 1 with respect to any action $a_1 \in \cA_1(k)$. Then $\vr_{1}^{(\thigh{k-2} )}[a_1] \leq - \sum_{l=2}^{k-2} (l-1) T_l$. Similarly, for any odd $\kappa + 2 \geq k \geq 5$, if $\vr_2^{(\thigh{k-2})}[a_2]$ is the regret of Player 2 with respect to any action $a_2 \in \cA_2(k)$, $\vr_{2}^{(\thigh{k-2} )}[a_2] \leq - \sum_{l=2}^{k-2} (l-1) T_l$.
\end{lemma}

\begin{proof}
    Let $a_1 \in \cA_1(k)$ and $l \in [k-2]$ for an even $k$. Playing $a_1 \in \cA_1(k)$ during $[\tlow{l}, \thigh{l}]$ gives Player 1 a utility of $0$; this follows from the fact that for any column $a_2 \in \{ a_2(1), a_2(3), \dots, a_2(k-3) \} = \{ a_2(1), a_2(2), a_2(3), \dots, a_2(k-3), a_2(k-2) \} $, it holds that $\mat{A}[a_1(k), a_2] = 0$, by construction of $\mat{A}$. At the same time, Player 1 actually got a utility of at least $l-1$ for each round in $[\tlow{l}, \thigh{l}]$. This means that every time Player 1 updates its regret vector within the time period $[\tlow{l}, \thigh{l}]$, the regret of $a_1$ decreases by at least $l-1$. The same reasoning applies for Player 2 when $k$ is odd.
\end{proof}

\begin{lemma}
    \label{lemma:indtwo}
    For any even $\kappa \geq k \geq 4$, $T_k \geq - \frac{1}{2} \vr_{2}^{(\thigh{k-1})}[a_2(k+1)]$. Similarly, for every odd $\kappa  \geq k \geq 5$, $T_k \geq - \frac{1}{2} \vr_{1}^{(\thigh{k-1})}[a_1(k+1)]$.
\end{lemma}

\begin{proof}
    We fix an even $k \geq 4$. $T_k$ is at least as large as the number of rounds it takes for $a_2(k+1)$ to have nonnegative regret, starting from $\vr_{2}^{(\thigh{k-1})}[a_2(k+1)]$. But in every round in $[\tlow{k}, \thigh{k}]$ the regret of $a_2(k+1)$ can increase additively by at most $2$. This holds because the utility of Player 2 for playing $a_2(k+1)$ is larger than the utility obtained in each round in $[\tlow{k}, \thigh{k}]$ by at most $2$. So, $T_k \geq \lceil - \frac{1}{2} \vr_{2}^{(\thigh{k-1})}[a_2(k+1)] \rceil \geq - \frac{1}{2} \vr_{2}^{(\thigh{k-1})}[a_2(k+1)]$. The same reasoning applies when $k$ is odd.
\end{proof}

The following upper bound on the regret is crude, but will suffice for our purposes.

\begin{lemma}[Regret upper bound]
    \label{lemma:reg-ub}
    For any even $\kappa \geq k \geq 4$, $\| [\vr_1^{(\thigh{k})}]^+ \|_\infty \leq  2 \| [\vr_1^{(\thigh{k-2})}]^+ \|_\infty + 2 \leq \frac{5}{3} 2^{k/2}$ since $\| \vr_1^{(\thigh{2})} \|_\infty \leq \frac{4}{3}$. Similarly, for any odd $k \geq 5$, $\| [\vr_2^{(\thigh{k})}]^+ \|_\infty \leq 2 \| [\vr_2^{(\thigh{k-2})}]^+ \|_\infty + 2 \leq \frac{5}{3} 2^{(k-1)/2}$ since $\| \vr_2^{(\thigh{3})} \|_\infty \leq \frac{4}{3} $.
\end{lemma}

\begin{proof}
    The fact that $\| \vr_1^{(\thigh{2})} \|_\infty, \| \vr_2^{(\thigh{3})} \|_\infty \leq \frac{4}{3}$ can be shown as part of the basis of the induction. We make the argument for an even $k$. From round $\tlow{k}$ until Player 1 plays $a_1(k)$ with probability 1, the regret of $a_1(k)$ increases by $k - (k \vx^{(t)}_1[a_1(k)] + (k-1) \vx^{(t)}_1[a_1(k-2)]) = \vx^{(t)}_1[a_1(k-2)]$ and the regret of $a_1(k-2)$ increases by $k-1 - (k \vx^{(t)}_1[a_1(k)] + (k-1) \vx^{(t)}_1[a_1(k-2)]) = -1 + \vx^{(t)}_1[a_1(k-2)]$; that is, it decreases by $1 - \vx^{(t)}_1[a_1(k-2)]$. Let $t'$ be the first round for which $\vr^{(t')}[a_1(k)] \geq \vr^{(t')}[a_1(k-2)]$. It holds that $\vr^{(t')}[a_1(k)] \leq \| [\vr_1^{(\thigh{k-2})}]^+ \|_\infty + 1 $ since the regret of $a_1(k)$ is increasing by at most $1$ in each round and $\vr^{(t')}[a_1(k-2)] \leq \| [\vr_1^{(\thigh{k-2})}]^+ \|_\infty$. From then onward, the regret of $a_1(k)$ is increasing by at most $1/2$ while the regret of $a_1(k-2)$ is decreasing by at least $1/2$. Thus, it will take at most $\lceil 2 | \vr^{(t')}[a_1(k-2)] | \rceil \leq 2 | \vr^{(t')}[a_1(k-2)] | + 1 \leq 2 \| [\vr_1^{(\thigh{k-2})}]^+ \|_\infty + 1$ rounds for the regret of $a_1(k-2)$ to be nonpositive. During that time, the regret of $a_1(k)$ can increase by at most $\| [\vr_1^{(\thigh{k-2})}]^+ \|_\infty + 1$. 
\end{proof}

\begin{proof}[Proof of~\Cref{lemma:resting-transition}]
    If $\kappa$ is odd, it suffices to prove that in every round in which Player 1 mixes between $a_1(\kappa)$ and $a_1(\kappa + 1)$, Player 2 plays $a_2(\kappa) = a_2(\kappa+1)$ with probability 1. Similarly, if $\kappa$ is even, it suffices to prove that in every round in which Player 2 mixes between $a_2(\kappa)$ and $a_2(\kappa+1)$, Player 1 plays $a_1(\kappa) = a_1(\kappa+1)$ with probability 1. Let us analyze the case where $\kappa$ is even; the odd case is similar. When Player 2 starts mixing more and more to $a_2(\kappa+1)$, it makes the row $a_1(\kappa + 2)$ more attractive for Player 2. By~\Cref{lemma:indone,lemma:indtwo}, 
    \begin{equation}
        \label{eq:reglow}
        \vr_1^{(\thigh{\kappa})}[a_1(\kappa+2)] \leq - \frac{\kappa - 1}{2} T_{\kappa} - \frac{\kappa - 2}{2} T_{\kappa-1} \leq - \frac{(\kappa - 1)!}{ 2^{\kappa - 2}} T_4 - \frac{(\kappa - 2)!}{ 2^{\kappa - 3}} T_3.
    \end{equation}
    At the same time, \Cref{lemma:reg-ub} implies that Player 2 is mixing between $a_2(\kappa)$ and $a_2(\kappa+1)$ for at most $3 \| [\vr_2^{(\thigh{\kappa-1})}]^+ \|_\infty + 2 \leq 5 \cdot 2^{(\kappa-1)/2} + 2$ rounds. To see this, we observe that it takes at most $\lceil \| [\vr_2^{(\thigh{\kappa-1})}]^+ \|_\infty \rceil $ rounds for the action $a_2(\kappa+1)$ to be played with at least the same probability as $a_2(\kappa)$, which in turn holds because the regret of $a_2(\kappa+1)$ increases by $\vx^{(t)}_2[a_2(\kappa)]$ while the regret of $a_2(\kappa)$ decreases by $1 - \vx_2^{(t)}[a_2(\kappa)]$. From then on, the regret of $a_2(\kappa)$ decreases by at least $1/2$ in each round, so it takes at most $\lceil 2 \| [\vr_2^{(\thigh{\kappa-1})}]^+ \|_\infty \rceil$ rounds for it to be nonpositive. We now claim that, by~\eqref{eq:reglow}, action $a_1(\kappa+2)$ is never played during those rounds. The reason is that since $T_3 \geq 5$ and $T_4 \geq 20$ (by our inductive basis),
    \begin{equation*}
        \vr_1^{(\thigh{\kappa})}[a_1(\kappa+2)]  \leq - 20 \frac{(\kappa - 1)!}{2^{\kappa - 2}} - 5 \frac{(\kappa-2)!}{2^{\kappa - 3}}
    \end{equation*}
    and in each round the regret of $a_1(\kappa+2)$ can only increase additively by $2$. Since
    \begin{equation*}
    \frac{1}{2} \left(  20 \frac{(\kappa - 1)!}{2^{\kappa - 2}} + 5 \frac{(\kappa-2)!}{2^{\kappa -3}} \right) > 5 \cdot 2^{(\kappa-1)/2} + 2 \quad \forall \kappa \geq 4,   
    \end{equation*}
    the inductive step follows.
\end{proof}

The next lemma shows that, under the invariance of~\Cref{lemma:resting-transition}, the only way to reach an approximate Nash equilibrium is to start playing the actions corresponding to $2m -1$, which is the maximum payoff in the matrix.

\begin{lemma}
    \label{lemma:mixedNE}
    Consider any strategy profile $(\vx_1, \vx_2)$ such that Player 1 only assigns positive probability to actions in $\{ a_1(k), a_1(k+1) \}$ and Player 2 only assigns positive probability to actions in $\{ a_2(k), a_2(k+1) \}$, where $k+1 < 2m - 1$. Then either Player 1 or Player 2 has a deviation benefit of at least $\nicefrac{1}{k + 2} - \gamma$ for any $\gamma > 0$.
\end{lemma}

\begin{proof}
    By construction of the game, either $a_1(k) = a_1(k+1)$ or $a_2(k) = a_2(k+1)$. We can assume that $a_1(k) = a_1(k+1)$; the argument when $a_2(k) = a_2(k+1)$ is symmetric. Let $p$ be the probability Player 2 places at $a_2(k+1)$ and $1 - p$ at $a_2(k)$. Suppose that the deviation benefit of each player is at most $\epsilon$. The utility of Player 2 under the current strategy profile is $k (1 - p) + (k+1) p = k + p$, while deviating to $a_2(k+1)$ gives $k+1$. So, $p > 1 - \epsilon$. Given that $k+1 < 2m - 1$, Player 1 can deviate to $a_1(k+2)$ to obtain a utility of $p (k+2) \leq k + p + \epsilon$. Combining with the fact that $p \geq 1 - \epsilon$, this implies $\epsilon \geq \nicefrac{1}{k+2}$.
\end{proof}

\paragraph{Uniform initialization}

So far, our lower bound assumes that one can initialize $\RM$ arbitrarily. A more common initialization prescribes randomizing uniformly at random. In what follows, we point out that the same argument works by considering a different common payoff matrix; namely,
\begin{equation*}
    \R^{(2m) \times (2m)} \ni \tilde{\mat{B}}[a_1, a_2] \defeq 
    \begin{cases}
        \mat{A}[a_1, a_2] & \text{if } 1 \leq a_1 \leq m \text{ and } 1 \leq a_2 \leq m;\\
        1 - \frac{1}{m} & \text{if } a_1 = 1 \text{ and } m+1 \leq a_2 \leq 2m ;\\
        1 - \frac{3}{m} & \text{if } m+1 \leq a_1 \leq 2m \text{ and } a_2 = 1;\\
        - \frac{1}{m} \sum_{a_2' = 1}^m \mat{A}[a_1, a_2'] & \text{if } 2 \leq a_1 \leq m \text{ and } m+1 \leq a_2 \leq 2m;\\
        - \frac{1}{m} \sum_{a_1' = 1}^m \mat{A}[a_1', a_2] & \text{if } m+1 \leq a_1 \leq 2m \text{ and } 2 \leq a_2 \leq m;\\
        \frac{1}{m^2} \sum_{a_1' = 1}^m \sum_{a_2' = 1}^m \mat{A}[a_1', a_2'] - \frac{2}{m} & \text{if } m+1 \leq a_1 \leq 2m \text{ and } m+1 \leq a_2 \leq 2m.
    \end{cases}
\end{equation*}
By assumption, $\vx_1^{(1)}$ and $\vx_2^{(1)}$ are the uniform distributions over $[2m]$. We observe that
\begin{align*}
    \sum_{a_1 = 1}^{2m} \sum_{a_2 = 1}^{2m} \tilde{\mat{B}}[a_1, a_2] &= \sum_{a_1 = 1}^{m} \sum_{a_2 = 1}^{m} \mat{A}[a_1, a_2] + m \left( 1 - \frac{1}{m} \right) + m \left( 1 - \frac{3}{m} \right) \\
    &- \sum_{a_1 = 2}^m \sum_{a_2 = 1}^m \mat{A}[a_1, a_2] - \sum_{a_1 = 1}^m \sum_{a_2 = 2}^m  \mat{A}[a_1, a_2] + \sum_{a_1 = 1}^m \sum_{a_2 = 1}^m \mat{A}[a_1, a_2] - 2m \\
    &= 2 \sum_{a_1 = 1}^{m} \sum_{a_2 = 1}^{m} \mat{A}[a_1, a_2] - \sum_{a_1 = 2}^m \sum_{a_2 = 1}^m \mat{A}[a_1, a_2] - 1 - \sum_{a_1 = 1}^m \sum_{a_2 = 2}^m  \mat{A}[a_1, a_2] - 3 = 0
\end{align*}
since $1 = \sum_{a_2 = 1}^m \mat{A}[a_1 = 1, a_2]$ and $3 = \sum_{a_1 = 1}^m \mat{A}[a_1, a_2=1]$. In turn, this implies that $\vx_1^{(1)} \tilde{\mat{B}} \vx_2^{(1)} = \frac{1}{4 m^2} \sum_{a_1 = 1}^{2m} \sum_{a_2 = 1}^{2m} \tilde{\mat{B}}[a_1, a_2] = 0$. Let $\vu_1^{(1)} = \tilde{\mat{B}} \vx_2^{(1)}$ and $\vu_2^{(1)} = \tilde{\mat{B}}^\top \vx_1^{(1)}$. We have
\begin{equation*}
    \vu_1^{(1)}[a_1] = 
    \begin{cases}
        \frac{1}{2m} \sum_{a_2 = 1}^m \mat{A}[a_1 = 1, a_2] + \frac{1}{2m} m \left( 1 - \frac{1}{m} \right) = \frac{1}{2} & \text{if } a_1 = 1;\\
        \frac{1}{2m} \sum_{a_2=1}^m \mat{A}[a_1, a_2] - \frac{1}{2m} m \frac{1}{m} \sum_{a_2 = 1}^m \mat{A}[a_1, a_2] = 0 & \text{if } 2 \leq a_1 \leq m;\\
        \frac{1}{2m} \left( 1 - 2 \right) = -\frac{1}{2m} & \text{if } m+1 \leq a_1 \leq 2m.
    \end{cases}
\end{equation*}
Since $\langle \vx_1^{(1)}, \vu_1^{(1)} \rangle = 0$, 
\begin{equation*}
    \vr_1^{(1)}[a_1] = 
    \begin{cases}
        \frac{1}{2} & \text{if } a_1 = 1;\\
        0 & \text{if } 2 \leq a_1 \leq m;\\
        -\frac{1}{2m} & \text{if } m+1 \leq a_1 \leq 2m.
    \end{cases}
\end{equation*}
Similarly, we have
\begin{equation*}
    \vr_2^{(1)}[a_2] = 
    \begin{cases}
        \frac{1}{2} & \text{if } a_2 = 1;\\
        0 & \text{if } 2 \leq a_2 \leq m;\\
        -\frac{1}{2m} & \text{if } m+1 \leq a_2 \leq 2m.
    \end{cases}
\end{equation*}
As a result, the second round sees both players playing the first action with probability 1. In view of the fact that $\mat{A}[a_1=1, a_2] < 1$ when $1 \leq a_2 \leq m$ and $\mat{A}[a_1, a_2 = 1] < 1$ when $1 \leq a_1 \leq m$, we immediately revert to the previous analysis concerning simultaneous $\RM$ executed on the game described in~\eqref{eq:gamelower}. Thus, we arrive at the following lower bound (by a change of variables $m' \leftarrow 2m$).

\begin{corollary}
    \label{cor:uniform-init}
    Simultaneous $\RM$ requires $m^{\Omega(m)}$ rounds to converge to a $\frac{1}{m + 1}$-Nash equilibrium in two-player $m \times m$ identical-interest games. This holds even when both players initialize at the uniform random strategy.
\end{corollary}

\end{document}